\documentclass[runningheads]{llncs}
\usepackage{graphicx}
\usepackage{babel}
\usepackage{amsmath}
\usepackage{amssymb}
\usepackage{comment}
\usepackage{tikz}
\usepackage{xspace}
\usepackage{colortbl}
\usepackage{subcaption}

\usetikzlibrary{arrows}

\newcommand{\set}[1]{\left\{ #1 \right\}}
\newcommand{\card}[1]{\left| #1 \right|}

\newcommand{\gjs}{H_{\mbox{\scriptsize{GJS}}}}

\newcommand{\C}{\mathcal{C}}
\newcommand{\A}{\mathcal{A}}

\renewcommand{\S}{\mathcal{S}}

\let\doendproof\endproof
\renewcommand\endproof{~\hfill$\qed$\doendproof}

\newcommand{\defdecisionproblem}[3]{
 \vspace{1mm}
\noindent\fbox{
 \begin{minipage}{0.96\textwidth}
{\textsc{#1}} \\
 {\bf{Input:}} #2 \\
 {\bf{Question:}} #3
 \end{minipage}
 }
 \vspace{1mm}
}

\newcommand{\defparproblem}[4]{
 \vspace{1mm}
\noindent\fbox{
 \begin{minipage}{0.96\textwidth}
 \begin{tabular*}{\textwidth}{@{\extracolsep{\fill}}lr} \textsc{#1} & {\bf{Parameter:}} #3 \\ \end{tabular*}
 {\bf{Input:}} #2 \\
 {\bf{Question:}} #4
 \end{minipage}
 }
 \vspace{1mm}
}

\newtheorem{observation}[theorem]{Observation}

\begin{document}
\title{$1$-Extendability of independent sets
\thanks{Partially
supported by the LABEX MILYON (ANR-10-LABX-0070) of Université de Lyon, within
the program ``Investissements d’Avenir'' (ANR-11-IDEX-0007), and the research grant ANR DIGRAPHS ANR-19-CE48-0013-01, operated by the French
National Research Agency (ANR).}}
%
%
\author{Pierre Bergé\orcidID{0000-0002-6170-9473} \and
Anthony Busson\orcidID{0000-0002-9445-637X} \and
Carl Feghali\orcidID{0000-0001-6727-7213} \and
Rémi Watrigant\orcidID{0000-0002-6243-5910}}
\authorrunning{P. Bergé et al.}
%
\institute{Univ Lyon, CNRS, ENS de Lyon, Université Claude Bernard Lyon 1, LIP UMR5668, France\\
\email{\{pierre.berge, anthony.busson, carl.feghali, remi.watrigant\}@ens-lyon.fr}}
\maketitle              
\begin{abstract}
In the 70s, Berge introduced 1-extendable graphs (also called B-graphs), which are graphs where every vertex belongs to a maximum independent set. 
Motivated by an application in the design of wireless networks, we study the computational complexity of \textsc{$1$-extendability}, the problem of deciding whether a graph is 1-extendable. 
We show that, in general, \textsc{$1$-extendability}  cannot be solved in $2^{o(n)}$ time assuming the Exponential Time Hypothesis, where $n$ is the number of vertices of the input graph, and that it remains NP-hard in subcubic planar graphs and in unit disk graphs (which is a natural model for wireless networks). 
Although \textsc{$1$-extendability} seems to be very close to the problem of finding an independent set of maximum size (\textit{a.k.a.} \textsc{Maximum Independent Set}), we show that, interestingly, there exist 1-extendable graphs for which \textsc{Maximum Independent Set} is NP-hard.
Finally, we investigate a parameterized version of \textsc{$1$-extendability}.

\end{abstract}

\section{Introduction and Motivation}

\subsection{Definitions and Related Work}
%
Understanding the structure of independent sets is among the most studied subjects in algorithmics and graph theory, and finding graph classes where a maximum independent set (MIS for short) can be found efficiently is an important theoretical and practical problem. 
In 1970, Plummer~\cite{Pl70} defined the class of well-covered graphs, which are graphs where every independent set which is maximal for inclusion is also an MIS. In other words, they are exactly the graphs for which the greedy algorithm always returns an optimal solution. 
Well-covered graphs were studied mostly from an algorithmic perspective: their recognition was proven coNP-hard~\cite{ChSl93,SaSt92} in general graphs, but polynomial-time solvable for claw-free graphs~\cite{TaTa96}, and perfect graphs of bounded clique number~\cite{DeZi94}.

A related notion, introduced by Berge~\cite{Be81}, is the definition of B-graphs, which are those graphs where every vertex belongs to an MIS. B-graphs were mostly introduced in order to study well-covered graphs~\cite{Rav76,Rav77}. 
Later, the notion of B-graphs was generalized to that of $k$-extendable graphs~\cite{DeZi94}: a graph is $k$-extendable, for a positive integer $k$, if every independent set of size (exactly) $k$ is contained in an MIS. Thus, B-graphs are exactly the 1-extendable graphs and a graph is well-covered if and only if (iff) it is $k$-extendable for every $k \in \{1, 2, \dots, \alpha(G)\}$, where $\alpha(G)$ denotes the size of an MIS of $G$. Dean and Zito~\cite{DeZi94} obtained a number of results on 1-extendable graphs; for instance, they proved that a bipartite graph is 1-extendable iff it admits a perfect matching and, hence, bipartite 1-extendable graphs can be recognized in polynomial time. Recently, certain structural properties of $k$-extendable graphs were stated~\cite{AnSuSw15,AnSuSw16}.

We should note that the notion of $k$-extendability was also studied in the context of maximum matchings~\cite{Pl80,Pl94}. Recently, it was shown that the recognition of (matching) $k$-extendable graphs is co-NP-complete~\cite{HaKo18}.

In the remainder, B-graphs will be called 1-extendable graphs, as it is the terminology used by the most recent papers on the topic. In this article, our objective is to determine the computational complexity of the recognition of 1-extendable graphs. This question is not only motivated by the state of the art described above but also by an application on Wi-Fi networks.




\subsection{CSMA/CA network and $1$-extendability}

Indeed, 1-extendable graphs play an important role in the performance of CSMA/ CA (Carrier Sense Multiple Access / Collision Avoidance) networks. 
CSMA/CA is the mechanism used by the nodes to access the radio channel in many wireless network technologies. 
It aims to prevent collisions, which happens when several nodes transmit at the same time thereby producing harmful interference that may cause transmissions losses. 
Basically, it is a listen-before-talk mechanism where a potential transmitter listens to the radio channel for a certain period of time, and transmits if the channel was sensed as idle during this period. 

Graphs stand as a natural model for CSMA/CA wireless networks. 
Two vertices, {\em i.e.} nodes of the CSMA/CA network, are adjacent if the two corresponding nodes are able to detect the transmissions from each others. 
Transmissions from two vertices can occur in parallel iff they are not adjacent. 
A set of instantaneous transmitters is thus an independent set of the graph.  

This graph, also named \textit{conflict graph} in the literature, is used to evaluate the network performance. 
The performance parameter that is often computed is the \textit{throughput} that offers to each vertex, {\em i.e.} 
the number of bits per second that a vertex is able to send.
The throughput of a vertex is strongly correlated to the proportion of time this vertex is transmitting. We denote by $p_v$ this quantity for the vertex $v \in V(G)$.
If we neglect the network protocol headers and transmission errors, the throughput of vertex $v$ may be considered as proportional to $p_v$. 
The first formal work that characterized $p_v$ has been developed in~\cite{csma-capacity}. 
It was shown that, under saturation condition, $p_v$ is given by: 

\begin{equation}
    p_v = \frac{\sum_{S \in \S(G) : v \in S}{\theta^{|S|}}}{\sum_{S \in \S(G)}\theta^{|S|}},
\end{equation}
where $\theta$ is the ratio between transmission and listen phase durations and $\S(G)$ is the collection of independent sets of $G$.
When $\theta$ tends to infinity, $p_v$ tends to the number of MISs of $G$ containing $v$ ($\#_v\alpha(G)$) divided by the total number $\#\alpha(G)$ of MISs of $G$:

\begin{equation}\label{eq:thr_limit}
    \lim_{\theta \rightarrow +\infty} p_v = \frac{\#_v\alpha(G)}{\#\alpha(G)}
\end{equation}

\begin{figure}[t]
\begin{subfigure}{0.5\textwidth}
  \centering
  \includegraphics[width=.9\linewidth]{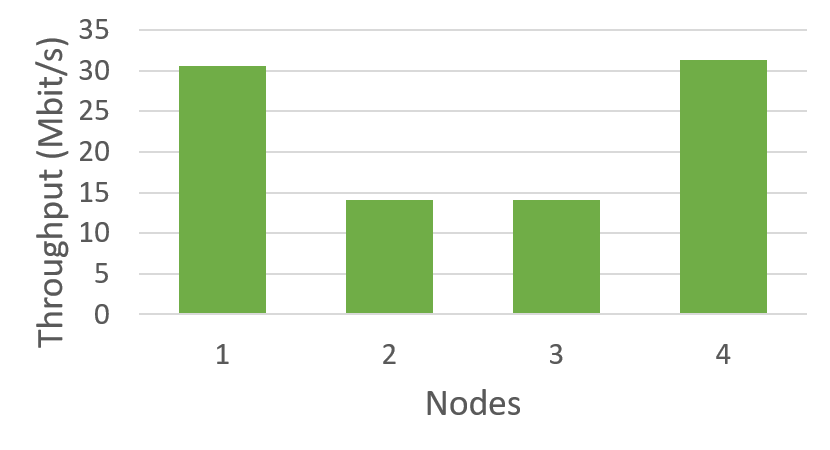}  
  \caption{4 vertices}
  \label{fig:throughput4nodes}
\end{subfigure}
\begin{subfigure}{0.5\textwidth}
  \centering
  \includegraphics[width=.8\linewidth]{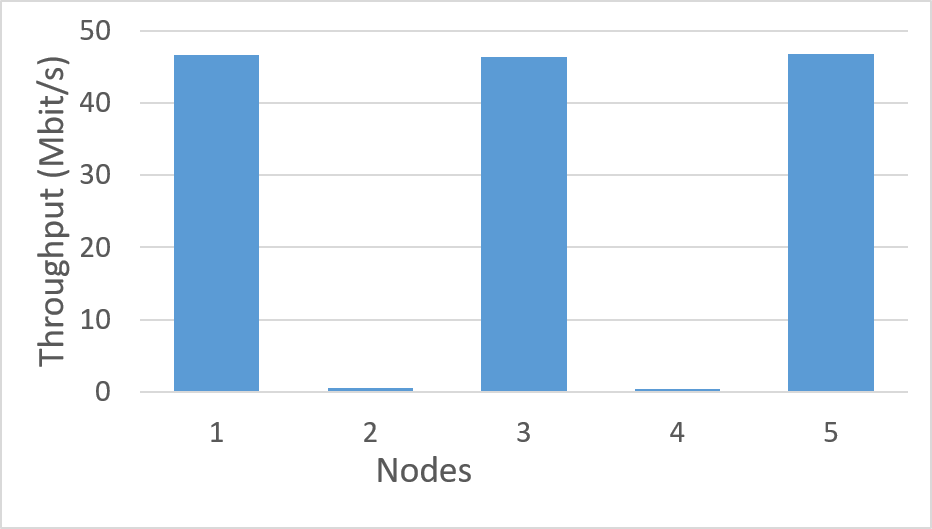}  
  \caption{5 vertices}
  \label{fig:throughput5nodes}
\end{subfigure}
\caption{ \small{Wi-Fi network simulated with ns-3 for paths on 4 and 5 vertices. 
The simulation parameters are: Wi-Fi 6 (IEEE 802.11ax) with a fixed MCS (He MCS=5), packet size=1024 bytes, aggregation is enable with a maximum of $16$ aggregated frame. The traffic is saturated.}}
\label{fig:throughput}
\end{figure}

In practice, the value of $\theta$ tends to be large to ensure an efficient channel use. 
In Wi-Fi networks for instance, typical values of $\theta$ then ranges from $20$ to $100$ depending on the transmission parameters.
With such values, a node/vertex that does not belong to any MIS will experience a very low throughput.
In Fig.~\ref{fig:throughput}, we show the throughput obtained in a Wi-Fi network for paths on $4$ and $5$ vertices using the network simulator ns-3~\cite{ns3}. 
The 5-vertex path is not 1-extendable, and we can observe that the two vertices that do not belong to any MIS are in starvation: they admit a very low throughput. 
In the 4-vertex path, there is no starvation as all vertices belong to at least one MIS. 

The fact, for each vertex, of belonging or not to an MIS is thus of prior importance to ensure a minimal fairness between the vertices and to avoid starvation. 
CSMA/CA networks where the parameter $\theta$ is large must thus be designed in such a way that the resulting conflict graph is 1-extendable. 

\subsection{Contribution}

Most of the outcomes of this paper concern the complexity of \textsc{$1$-Extendability}, the problem of deciding whether an input graph is 1-extendable. First, we focus on the relationship between \textsc{maximum independent set} and \textsc{$1$-Extendability}. We observe that any polynomial-time algorithm for \textsc{maximum independent set} on some hereditary family of graphs $\mathcal{C}$ provides us with a polynomial-time algorithm for \textsc{$1$-Extendability} on $\mathcal{C}$. Based on this result, we could imagine that, perhaps, \textsc{maximum independent set} and \textsc{$1$-Extendability} are equivalent problems in terms of complexity. However, we show that \textsc{Maximum Independent Set} is NP-hard on a certain subfamily of 1-extendable graphs (Theorem~\ref{th:gap}). This result highlights a gap for the complexity of these two problems. 

We provide a linear reduction which implies, under the Exponential Time Hypothesis (ETH) that \textsc{$1$-Extendability} cannot be solved in time $2^{o(n)}$ on an $n$-vertex input graph (Corollary~\ref{co:eth-hard}). 

Second, we establish the NP-hardness of \textsc{$1$-Extendability} on certain families of graphs. We prove that the problem is NP-hard in planar subcubic graphs (Theorem~\ref{th:hard_planar}) and in unit disk graphs (Theorem~\ref{th:unit_disk}), a natural model for CSMA/CA networks.

Eventually, we focus on a parameterized version of \textsc{$1$-Extendability}, where we ask whether every vertex belongs to an independent set of size at least some parameter $k$. We show that this problem, \textsc{param-$1$-Extendability}, is W[1]-hard (Theorem~\ref{thm:w1}). Nevertheless, it admits a polynomial kernel if restricted to planar graphs or $K_r$-free graphs for fixed $r > 0$ (Corollary~\ref{co:kernel}).

\subsection{Organization of the paper}

Section~\ref{sec:prelim} is dedicated to the notation, definitions and some basic results; in particular, we explore the relationship between the \textsc{Maximum Independent Set} and \textsc{$1$-Extendability} problems. 
In Section~\ref{sec:transfo}, we study three graph transfor- -mations and their impact on the $1$-extendability property.
In Section~\ref{sec:hardness}, we show that \textsc{$1$-Extendability} cannot be solved in time $2^{o(n)}$ on $n$-vertex graphs unless the ETH is false. We also prove that \textsc{$1$-Extendability} remains NP-hard in planar graphs of maximum degree $3$ and in unit disk graphs.
Then, Section~\ref{sec:param} presents the parameterized version \textsc{param-$1$-Extendability} and the results associated with it.
We conclude and give several open questions in Section~\ref{sec:conclusion}.


\section{Notation and Basic Results} \label{sec:prelim}


\subsection{Notation and definitions} \label{subsec:notation}

For a positive integer $k$, we write $[k] = \{1, \dots, k\}$. In this paper, all graphs are simple, unweighted and undirected. We denote by $V(G)$ the vertex set of a graph $G$ and by $E(G)$ its edge set. Edges $(u,v) \in E(G)$ can sometimes be denoted by $uv$ to improve readability. When the identity of the graph considered is clear, we set $n = \card{V(G)}$ and $m = \card{E(G)}$. 
For a vertex $v \in V(G)$, we denote by $N_G(v)$ its set of neighbors (we will sometimes omit the subscript if $G$ is clear from the context). Let $d(v)$ be the degree of $v$, {\em i.e.} $d(v) = \card{N(v)}$.
For a subset $R \subseteq V(G)$, let $G[R]$ be the subgraph of $G$ induced by $R$.  A family of graphs $\C$ is called \textit{hereditary} if, for every graph $G \in \C$, every induced subgraph of $G$ also belongs to $\C$.
An $\ell$-vertex path is denoted by $P_{\ell}$.
A \textit{clique cover} of a graph $G$ is a partition of $V(G)$ into sets $C_1$, $\dots$, $C_q$ such that $G[C_i]$ is a clique for every $i \in [q]$.
 A set of pairwise non-adjacent vertices in a graph is called an \emph{independent set}. A \textit{maximum} independent set (MIS) is an independent set of maximum size. We denote by $\alpha(G)$ the size of an MIS of $G$. The decision problem of finding an independent set of size at least $k\ge 1$ in a graph $G$ is called \textsc{Maximum Independent Set}. A graph $G$ is 1-\textit{extendable}~\cite{Be81} if, for every  $u \in V(G)$, there is an MIS $S$ of $G$ such that $u \in S$. 

The subject of this paper is to investigate the computational complexity of the following decision problem. 

\medskip

\defdecisionproblem{$1$-Extendability}{Graph $G$}{Does every vertex of $G$ belong to an MIS of $G$?}

\medskip

An \textit{embedding} of a graph $G$ is a representation of $G$ in the plane, where vertices are points in the plane and edges are curves which connect their two endpoints. A \textit{plane embedding} of $G$ is an embedding of $G$ where no two edges cross. A graph $G$ is \textit{planar} if it admits a plane embedding.
A \textit{parameterized problem} is a decision problem where an integer (called the parameter) is associated to every instance.
A Fixed-Parameter Tractable (FPT) algorithm is an algorithm deciding whether an instance of a parameterized problem is positive in time $f(k)n^{O(1)}$, where $f$ is a computable function, $n$ is the size of the instance, and $k$ is the parameter of the instance.
The $W$-hierarchy allows to rule out the existence of FPT algorithms for some problems: proving that a parameterized problem is $W[1]$-hard implies that it is unlikely to admit an FPT algorithm.
A \textit{kernel} for a parameterized problem is a polynomial-time algorithm which transforms an instance $x$ with parameter $k$ into an instance $x'$ with parameter $k'$ such that (i) $x$ is positive iff $x'$ is positive (ii) $k' \leqslant k$, and (iii) $|x'| \leqslant f(k)$, where $f$ is a computable function called the \textit{size} of the kernel. If $f$ is a polynomial, we say that it is a \textit{polynomial kernel}.
For more details about parameterized algorithms, we refer the reader to the textbook of Cygan et al.~\cite{CyFoKoLoMaPiPiSa15}.

\subsection{Links between \textsc{$1$-Extendability} and \textsc{Maximum Independent Set}} \label{sec:links}

In this section, we investigate to what extent the problems \textsc{$1$-Extendability} and \textsc{Maximum Independent Set} are close to each other. We show a ``Turing equivalence'' of the two problems in general graphs. More precisely, we prove that solving \textsc{$1$-Extendability} on an input graph $G$ can be done by solving \textsc{Maximum Independent Set} on several induced subgraphs of $G$, while solving \textsc{Maximum Independent Set} on an input graph $G$ can be done by solving \textsc{$1$-Extendability} on several induced supergraphs of $G$.

\paragraph{Solving \textsc{$1$-extendability} using \textsc{Maximum Independent Set}.}

The idea relies on the following lemma, whose straightforward proof is left to the reader. 
\begin{lemma}\label{lem:1ext-to-alpha}
Let $G$ be a graph, and $k$ be a non-negative integer. Then a vertex $v$ of $G$ is contained in an independent set of $G$ of size $k$ iff $G[V(G) \setminus N(v)]$ contains an independent set of size $k$.
\end{lemma}

This lemma allows \textsc{$1$-extendability} to inherit many positive results from \textsc{Maximum Independent Set}. In particular, it implies that \textsc{$1$-extendability} is polynomial-time solvable in any hereditary class  where \textsc{Maximum Independent Set} is polynomial-time solvable. This is for instance the case for perfect graphs, $P_6$-free graphs~\cite{GrKlPiPi19}, chordal graphs and claw-free graphs. Moreover, it is Fixed-Parameter Tractable (FPT) when parameterized by the tree-width or even the clique-width of the input graph. As we will see later in Section~\ref{sec:param}, this lemma can also be used to transfer more positive FPT results for a parameterized version of \textsc{$1$-extendability}.

\paragraph{Solving \textsc{Maximum Independent Set} using \textsc{$1$-extendability}.}

The converse of Lemma \ref{lem:1ext-to-alpha} does not appear to be as straightforward, and we leave as open whether solving \textsc{$1$-extendability} in a hereditary graph class $\mathcal{C}$ in polynomial-time allows one to solve \textsc{Maximum Independent Set} in $\mathcal{C}$ in polynomial-time. We can show, however, that this is true if the class satisfies much more conditions than just being hereditary. 

  Let $G = (V, E)$ be a graph and $r \leq |V|$ be a non-negative integer. Let $G^+_r$ be the graph obtained from $G$ by adding
  
  \begin{itemize}
  \item an independent set $S$ of size $|V|-r$ to $G$, 
  \item for each vertex $v$ of $G$, a new vertex $\pi_v$ adjacent to $v$ only, and 
  \item all possible edges between $S$ and the set $T = \{\pi_v: v \in V\}$. 
\end{itemize} 


\begin{proposition}
$G^+_{r}$ is $1$-extendable iff $\alpha(G) = r$.
\label{prop:turing_reduction}
\end{proposition} 
\begin{proof}
Let $n=|V(G)|$. First observe that $\alpha(G^+_r) \geq n$, since $T$ is an independent set. More precisely, by construction $\alpha(G^+_r) = \max\{n, n-r+\alpha(G)\}$.

Suppose that $G^+_r$ is $1$-extendable. By definition, every vertex of $S$ belongs to a MIS, and since all vertices of $S$ have the same neighborhood, there must be a MIS $I$ containing all vertices of $S$. Now, since the neighborhood of any vertex of $S$ is $T$, it follows that $G$ must contain an independent set of size $\alpha(G^+_r) - |I| \geq r$.  But if $G$ contains an independent set of size $r+1$, then $\alpha(G^+_r) \geq n+1$. However, for $v \in V(G)$, the non-neighborhood of $\pi_v$ is of size $n-1$, so $\pi_v$ cannot belong to an independent set of size $n+1$, contradicting the $1$-extendability of $G^+_r$. 
 
 Conversely, if $\alpha(G)=r$, then necessarily $\alpha(G^+_r) = n$. Let $J$ be an independent set of size $r$ in $G$.  Observe that:
 \begin{itemize}
 	\item for each vertex $v \in V(G)$, the set $\{v\} \cup \{\pi_u: u \in V(G) \setminus \{v\}\}$ is an independent set of size $n$;
 	\item $T$ is an independent set of size $n$;
 	\item $S \cup J$ is an independent set of size $n$.
 \end{itemize}
 In brief, each vertex of $G^+_r$ is contained in a MIS, which concludes the proof.
\end{proof}

The ETH~\cite{ImPa01} states that no algorithm can decide whether a 3SAT formula on $n$ variables is satisfiable in time $2^{o(n)}$. As 3-SAT, \textsc{maximum independent set} is ETH-hard. Hence, based on Proposition~\ref{prop:turing_reduction}, the same statement holds for \textsc{$1$-Extendability}.

\begin{corollary}
Testing whether an $n$-vertex graph is $1$-extendable cannot be done in time $2^{o(n)}$ unless the ETH is false.
\label{co:eth-hard}
\end{corollary}

This lower bound is matched by the trivial brute-force algorithm which consists in enumerating all subsets of vertices, and testing whether all MISs cover the entire vertex set.

Another question related to the previous one is whether being 1-extendable helps finding a MIS. The next result suggests that this is unlikely, by showing that \textsc{$1$-extendability} and \textsc{Maximum Independent Set} can sometimes behave very differently from a computational point of view. 

Let $k \geq 1$ be an integer. The \textsc{$k$-Multicolored Independent Set} problem, asks, given a graph $G$ whose vertex set can be partitioned into $k$ parts $C_1, \dots, C_k$ each inducing a clique, whether $G$ contains an independent set $I$ such that $|I \cap C_i| = 1$ for each $i \in \{1, \dots, k\}$. 

%
%




\begin{theorem}
\textsc{Maximum Independent Set} remains $NP$-hard and $W[1]$-hard (parameterized by $k$) in $1$-extendable graphs.
\label{th:gap}
\end{theorem}
\begin{proof}
We reduce from \textsc{$k$-Multicolored Independent Set} which is well-known to be $W[1]$-hard \cite{CyFoKoLoMaPiPiSa15}.  

Let $G$ be an instance of \textsc{$k$-Multicolored Independent Set} and let $C_1, \dots, C_k$ be its $k$ cliques. We construct a 1-extendable graph $H$ from $G$ such that $G$ contains an independent set intersecting each $C_i$ iff $H$ contains an independent set of size $2k$.  

To construct $H$, we take two copies $G_1$ and $G_2$ of $G$ and add two new sets of vertices $P^1 = \{\pi_1^1, \dots, \pi_1^k\}$ and $P^2 = \{\pi_2^1, \dots, \pi_2^k\}$. We then add all possible edges between $P^1$ and $P^2$ and, for  $i \in \{1, 2\}$ and $j \in \{1, \dots, k\}$, we make $\pi_i^j$ adjacent to each vertex of $C_i^j$, where  $C_i^j$ denotes the $j$th clique of $G_i$. This completes the construction. 

To see that $H$ is $1$-extendable, let $A_1$ and $A_2$ be maximum independents in $G_1$ and $G_2$, respectively, and thus of size at most $k$. Note by construction that $\alpha(H) = k+\alpha(G)$. Thus, by construction again, we have that
\begin{itemize} 
	\item for each $C_i^j$ and $x \in C_i^j$, the set $\{x\} \cup \{\pi_i^j : j \neq i\} \cup A_{3 - i}$ is independent and of size $\alpha(H)$, and
	\item for each $i \in \{1, 2\}$, the set $P^i \cup A_{3 - i}$ is independent and of size $\alpha(H)$, 
\end{itemize}
and hence $H$ is $1$-extendable as needed. 

Now, suppose $G$ and thus $G_1$ has an independent set $S$ of size $k$.  Then $P^2 \cup S$ is an independent set of size $2k$ in $H$, as claimed. 

Conversely, suppose $H$ contains an independent set $T$ of size $2k$. Let $F_1 = H[V(G_1) \cup P^1]$ and $F_2 = H[V(G_2) \cup P^2]$. Thus, $H = F_1 \cup F_2$. By construction, each $F_i$ has  independence number at most $k$, and thus $T$ intersects each $F_i$ on exactly $k$ vertices. Hence, since there are all possible edges between $P^1$ and $P^2$, $T$ must intersect either $G_1$ or $G_2$ on $k$ vertices, which in turn implies $G$ has an independent set of size $k$, as required. This completes the proof. 
\end{proof}

\section{Generic transformations} \label{sec:transfo}

In this subsection, we present three graph transformations. They are related in some sense to the $1$-extendability property: the first one produces a 1-extendable graph, the second one preserves the $1$-extendability of the input graph and the third one decreases the maximum degree of the input graph and keeps it 1-extendable. These transformations (or similar ideas) will be used later in some reductions.

Given any graph $G$ on $n$ vertices, the transformation~$(T_1)$ returns a graph $G_{(1)}$ with $2n$ vertices which is not only 1-extendable but also admits $G$ as an induced subgraph. The transformation $(T_2)$ consists in $2$-subdividing the edges of the graph. This operation preserves the $1$-extendability: $G$ is 1-extendable iff $G_{(2)}$ is too. Eventually, the transformation $(T_3)$ produces a graph $G_{(3)}$ with maximum degree 3 which is 1-extendable if $G$ is 1-extendable (note that the converse is not necessarily true). Transformations $(T_2)$ and $(T_3)$ are well-known, and provide a useful tool to prove hardness on some restricted graph classes. Furthermore, they preserve the planarity of the input graph.

\medskip

\textbf{Transformation} $(T_1)$. The graph $G_{(1)}$ is obtained from $G$ by adding a pendant vertex $\pi_u$ for any $u \in V(G)$. The vertex $\pi_u$ has degree one and is adjacent to $u$. The graph $G_{(1)}$ has thus $2n$ vertices and $m+n$ edges. One of its MISs is the set of pendant vertices $\set{\pi_u: u \in V(G)}$: we thus have $\alpha(G_{(1)}) = n = \card{V(G_{(1)})}/2$. This provides us with a trivial linear-size vertex-addition scheme to obtain a 1-extendable graph.

\begin{lemma}
For any graph $G$, $G_{(1)}$ is $1$-extendable.
\label{le:t1}
\end{lemma}
\begin{proof}
The graph $G_{(1)}$ admits a clique cover of size $n = \card{V(G_{(1)})}/2$ consisting of all edges $u\pi_u$. As a consequence, the size of an MIS of $G_{(1)}$ is at most $n$. Furthermore, the set of pendant vertices form an independent set of size $n$.
We prove now that each vertex of $G_{(1)}$ belongs to an independent set of size $n$. We know it is already the case for pendant vertices. Let $u$ be a non-pendant vertex of $G_{(1)}$. We fix the following $n$-sized set: $X_u = \set{u} \cup \set{\pi_v : v\neq u}$. All pendants are pairwise non-adjacent. Moreover, the neighborhood of $\set{\pi_v : v\neq u}$ contains exactly all non-pendant vertices, except $u$. Hence $X_u$ is independent. In brief, every non-pendant vertex $u$ also belongs to an MIS of $G_{(1)}$. 
\end{proof}

\medskip

\textbf{Transformation} $(T_2)$. The graph $G_{(2)}$ is obtained from $G$ by subdividing each of its edges an even number of times, {\em i.e.} each edge becomes an induced $P_{2\ell}$. In fact, $(T_2)$ is a well-known graph transformation which provides, for instance, the proof that \textsc{Maximum Independent set} remains $NP$-hard on graphs forbidding a fixed graph $H$ as an induced subgraph, for any $H$ different from a path or a subdivided claw~\cite{Al82,Po74}.
This transformation preserves in some sense all independent sets of the input graph $G$. 
\begin{observation}[\cite{Al82,Po74}]
Consider any MIS $X'$ of $G_{(2)}$ and pick all its vertices $X \subsetneq X'$ which were already present in $G$, {\em i.e.} which do not belong to the subdivided edges. Then $X$ is an MIS of the input graph $G$. Additionally, the set $X'\backslash X$ contains exactly half of the vertices formed by the subdivisions.
\label{obs:subdivide}
\end{observation}
One can see that $G_{(2)}$ is planar iff $G$ is planar (subdivisions do not influence planar embeddings). Moreover, $(T_2)$ also preserves the $1$-extendability of the input graph.

\begin{lemma}
$G$ is 1-extendable iff $G_{(2)}$ is 1-extendable.
\label{le:t2}
\end{lemma}
\begin{proof}
We begin with some notation. For every vertex $u\in V(G)$, we denote by $\mathcal{R}[u]$ the set of vertices of $G_{(2)}$ which are (i) part of a subdivided edge incident to $u$ and (ii) at even distance from $u$. We consider $u \in \mathcal{R}[u]$ as it is at distance zero from itself. We claim that $u$ belongs to an MIS of $G$ iff $\mathcal{R}[u]$ is a subset of an MIS of $G_{(2)}$.

If $u$ is isolated, then it naturally belongs to all MISs of $G$. It stays isolated in $G_{(2)}$ and $\mathcal{R}[u] = \set{u}$, so our claim holds.

Assume that $\card{\mathcal{R}[u]}\ge 2$, {\em i.e.} $u$ has at least one neighbor in $G$. On the one hand, let $X$ be an MIS of $G$ containing $u$. Let $\mathcal{R}[X]= \bigcup_{v \in X} \mathcal{R}[v]$. As $X$ does not contain two adjacent vertices of $G$, then no adjacency appears in $\mathcal{R}[X]$. According to Observation~\ref{obs:subdivide}, $\mathcal{R}[X]$ is an MIS of $G_{(2)}$ and $\mathcal{R}[u] \subseteq \mathcal{R}[X]$. On the other hand, let $X'$ be an MIS of $G_{(2)}$ such that $\mathcal{R}[u] \subseteq X'$. According to Observation~\ref{obs:subdivide}, as $u \in \mathcal{R}[u]$, there is an MIS of $G$ containing $u$.

We can now prove that $(T_2)$ preserves $1$-extendability. If $G$ is 1-extendable, then every set $\mathcal{R}[u]$, $u \in V(G)$, is a subset of some MIS of $G_{(2)}$. Observe that $V(G_{(2)}) = \bigcup_{u \in V(G)} \mathcal{R}[u]$, hence $G_{(2)}$ is 1-extendable. Conversely, if $G_{(2)}$ is 1-extendable, then every vertex of the original graph $G$ appears within an MIS of $G_{(2)}$. By Observation~\ref{obs:subdivide}, it belongs to at least one MIS of $G$.
\end{proof}

\textbf{Transformation} $(T_3)$. The graph $G_{(3)}$ is obtained from $G$ by replacing each of its vertices by a path in order to decrease the maximum degree of the graph. It is a folklore transformation which also works for other classical problems.

First, we replace each vertex $u \in V(G)$ by an induced odd path $P_u$ of length $\ell = 2\Delta-1$, where $\Delta$ is the maximum degree of $G$. We denote by $u_1,\ldots,u_{\ell}$ the vertices of $P_u$. The vertex set of $G_{(3)}$ is $V(G_{(3)}) = \set{u_1,\ldots,u_{\ell}: u \in V(G)}$. Second, let $Q_u \subseteq P_u$ be the set of vertices in $P_u$ with odd index: $Q_u = \{u_{2i+1}: 0\le i\le \Delta-1\}$. For any $1\le i\le d(u)$, we assign arbitrarily to each vertex $u_{2i+1}$ of $Q_u$ a neighbor $\rho(u_{2i+1}) \in V(G)$ of $u$, so that $\rho$ is bijective. There are two types of edges in $G_{(3)}$:
\begin{itemize}
\item edges of induced paths $P_u$, $u\in V(G)$,
\item edges $u_{2i+1}v_{2j+1}$ when $\rho(u_{2i+1}) = v$ and $\rho(v_{2j+1}) = u$.
\end{itemize} 
In this way, the maximum degree $G_{(3)}$ is at most $3$.

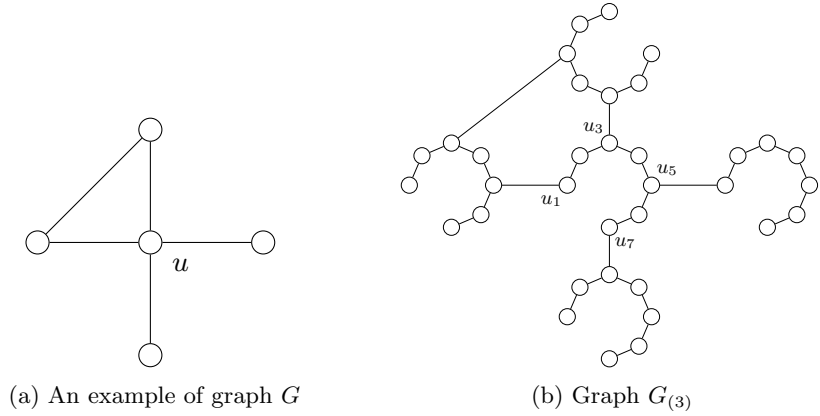
\begin{figure}
\begin{subfigure}[b]{0.49\columnwidth}
\centering
\scalebox{1.0}{\begin{tikzpicture}
	\node[draw,circle] (u) at (4.0,4.0) {};
	\node[draw,circle] (a) at (2.5,4.0) {};
	\node[draw,circle] (b) at (4.0,5.5) {};
        \node[scale = 1.2] (v) at (4.4,3.7) {$u$};
    \node[draw,circle] (c) at (5.5,4.0) {};
	\node[draw,circle] (d) at (4.0,2.5) {};
	
	\draw (u)--(a);
	\draw (u)--(b);
	\draw (u)--(c);
	\draw (u)--(d);
	\draw (a)--(b);	
\end{tikzpicture}}
\caption{An example of graph $G$}
\label{subfig:before_T3}
\end{subfigure}
\begin{subfigure}[b]{0.49\columnwidth}
\centering
\scalebox{0.7}{\begin{tikzpicture}

	\def \smval {0.4}
	
	\node[draw,circle] (u1) at (4.0-2*\smval,4.0) {};
		 \node[scale = 1.2] (v1) at (4.0-2*\smval-0.3,3.7) {$u_1$};
	\node[draw,circle] (u2) at (4.0-1.4*\smval,4.0+1.4*\smval) {};
	\node[draw,circle] (u3) at (4.0,4.0+2*\smval) {};
		\node[scale = 1.2] (v1) at (3.7,4.0+2*\smval+0.3) {$u_3$};
	\node[draw,circle] (u4) at (4.0+1.4*\smval,4.0+1.4*\smval) {};
	\node[draw,circle] (u5) at (4.0+2*\smval,4.0) {};
		\node[scale = 1.2] (v5) at (4.0+2*\smval+0.3,4.3) {$u_5$};
	\node[draw,circle] (u6) at (4.0+1.4*\smval,4.0-1.4*\smval) {};
	\node[draw,circle] (u7) at (4.0,4.0-2*\smval) {};
		\node[scale = 1.2] (v1) at (4.3,4.0-2*\smval-0.3) {$u_7$};

	\node[draw,circle] (a1) at (1.0-2*\smval,4.0) {};
	\node[draw,circle] (a2) at (1.0-1.4*\smval,4.0+1.4*\smval) {};
	\node[draw,circle] (a3) at (1.0,4.0+2*\smval) {};
	\node[draw,circle] (a4) at (1.0+1.4*\smval,4.0+1.4*\smval) {};
	\node[draw,circle] (a5) at (1.0+2*\smval,4.0) {};
	\node[draw,circle] (a6) at (1.0+1.4*\smval,4.0-1.4*\smval) {};
	\node[draw,circle] (a7) at (1.0,4.0-2*\smval) {};
	
	\node[draw,circle] (b1) at (4.0-2*\smval,6.5) {};
	\node[draw,circle] (b2) at (4.0-1.4*\smval,6.5+1.4*\smval) {};
	\node[draw,circle] (b3) at (4.0,6.5+2*\smval) {};
	\node[draw,circle] (b5) at (4.0+2*\smval,6.5) {};
	\node[draw,circle] (b6) at (4.0+1.4*\smval,6.5-1.4*\smval) {};
	\node[draw,circle] (b7) at (4.0,6.5-2*\smval) {};
	\node[draw,circle] (b8) at (4.0-1.4*\smval,6.5-1.4*\smval) {};
       
    \node[draw,circle] (c1) at (7.0-2*\smval,4.0) {};
	\node[draw,circle] (c2) at (7.0-1.4*\smval,4.0+1.4*\smval) {};
	\node[draw,circle] (c3) at (7.0,4.0+2*\smval) {};
	\node[draw,circle] (c4) at (7.0+1.4*\smval,4.0+1.4*\smval) {};
	\node[draw,circle] (c5) at (7.0+2*\smval,4.0) {};
	\node[draw,circle] (c6) at (7.0+1.4*\smval,4.0-1.4*\smval) {};
	\node[draw,circle] (c7) at (7.0,4.0-2*\smval) {};
        
	\node[draw,circle] (d1) at (4.0-2*\smval,1.5) {};
	\node[draw,circle] (d2) at (4.0-1.4*\smval,1.5+1.4*\smval) {};
	\node[draw,circle] (d3) at (4.0,1.5+2*\smval) {};
	\node[draw,circle] (d4) at (4.0+1.4*\smval,1.5+1.4*\smval) {};
	\node[draw,circle] (d5) at (4.0+2*\smval,1.5) {};
	\node[draw,circle] (d6) at (4.0+1.4*\smval,1.5-1.4*\smval) {};
	\node[draw,circle] (d7) at (4.0,1.5-2*\smval) {};
	
	\draw (u1)--(u2);
	\draw (u2)--(u3);
	\draw (u3)--(u4);
	\draw (u4)--(u5);
	\draw (u5)--(u6);
	\draw (u6)--(u7);
	
	\draw (a1)--(a2);
	\draw (a2)--(a3);
	\draw (a3)--(a4);
	\draw (a4)--(a5);
	\draw (a5)--(a6);
	\draw (a6)--(a7);
	
	\draw (b1)--(b2);
	\draw (b2)--(b3);
	\draw (b5)--(b6);
	\draw (b6)--(b7);
	\draw (b7)--(b8);
	\draw (b8)--(b1);
	
	\draw (c1)--(c2);
	\draw (c2)--(c3);
	\draw (c3)--(c4);
	\draw (c4)--(c5);
	\draw (c5)--(c6);
	\draw (c6)--(c7);
	
	\draw (d1)--(d2);
	\draw (d2)--(d3);
	\draw (d3)--(d4);
	\draw (d4)--(d5);
	\draw (d5)--(d6);
	\draw (d6)--(d7);
	
	\draw (u1)--(a5);
	\draw (u3)--(b7);
	\draw (u5)--(c1);
	\draw (u7)--(d3);
	\draw (a3)--(b1);	
\end{tikzpicture}}
\caption{Graph $G_{(3)}$}
\label{subfig:after_T3}
\end{subfigure}
\caption{Transformation $(T_3)$ embedded in such a way that planarity holds.}
\label{fig:T3}
\end{figure}

One may observe that $Q_u$ is an independent set of $P_u$ of maximum size $\Delta$. This is the key property which allows us to show that this structure maintains the $1$-extendability of the input graph. 

\begin{lemma}
Let $n = |V(G)|$. We have $\alpha(G_{(3)}) = n(\Delta-1)+\alpha(G)$.
Moreover, if $G$ is 1-extendable, then $G_{(3)}$ is 1-extendable.
\label{le:t3}
\end{lemma}
\begin{proof}
Let $X$ be an an MIS of $G$. We construct the following set in $G_{(3)}$: $$ X' = (\bigcup_{u \in X} Q_u)  \cup (\bigcup_{u\notin X} P_u\backslash Q_u) $$
On the one hand, no vertex of $\bigcup_{u\notin X} P_u\backslash Q_u$ has a neighbor in $X'$. On the other hand, as $X$ is an independent set, two vertices belonging respectively to $Q_u$ and $Q_v$, $u,v \in X$, are not adjacent. Therefore, $X'$ is an independent set of $G_{(3)}$ of size  $n(\Delta-1)+\alpha(G)$. 
Now, assume that there is an independent set $X^*$ of size at least $n(\Delta-1)+\alpha(G)+1$ in $G_{(3)}$. There are necessarily $\alpha(G)+1$ paths $P_u$ of $G_{(3)}$ which contain $\Delta$ vertices of $X^*$. The vertices $u \in V(G)$ satisfying this property must be pairwise non-adjacent, by definition of the edges $u_{2i+1}v_{2j+1}$. This yields a contradiction as we identify an independent set of $G$ of size $\alpha(G)+1$. 

We prove the second part of the statement, by proving that every vertex $w \in G_{(3)}$ belongs to an MIS. 
If $w$ is isolated, it belongs trivially to all MISs. 
If $w = u_{2i+1}$, $1\le i\le \Delta-1$, select an arbitrary MIS $X_u$ of $G$ containing $u$. We know that the set $X_u' = (\bigcup_{v \in X_u} Q_v)  \cup (\bigcup_{v\notin X_u} P_v\backslash Q_v)$ is an MIS of $G_{(3)}$ and $w \in X_u'$. If $w = u_{2i}$ and is not isolated, then we select an arbitrary MIS $Y_u$ of $G$ containing one of its neighbors. The set $Y_u' = (\bigcup_{v \in Y_u} Q_v)  \cup (\bigcup_{v\notin Y_u} P_v\backslash Q_v)$ is an MIS of $G_{(3)}$ and $w \in Y_u'$.
\end{proof}

Assume graph $G$ is planar. One can, by defining function $\rho$ in a good way (Fig.~\ref{fig:T3}), produce a graph $G_{(3)}$ which is still planar, according to~\cite{Mo01}. Unfortunately, one can find examples of graphs $G$ such that $G_{(3)}$ is 1-extendable while $G$ is not. In other words, unlike $(T_2)$, transformation $(T_3)$ does not produce an ``equivalent'' graph in terms of $1$-extendability.

\section{Hardness of \textsc{$1$-Extendability} on subcubic planar graphs}\label{sec:hardness}

The main goal of this section is to study the computational hardness of \textsc{$1$-Extendability}. We show that the problem is NP-hard in subcubic planar graphs and unit disk graphs.

\subsection{Properties of the Garey-Johnson-Stockmeyer gadget} \label{subsec:gadget}

We now focus on restricted graph classes. Since our first motivation is the context of wireless networks, we investigate the complexity of the problem in graphs modeling this kind of practical situations. Unit disk graphs is a natural graph class representing the conflict graph of wireless access points. 
As it is often the case when dealing with unit disk graphs, we first tackle the case of planar graphs of bounded degree.
There exists a well-known gadget~\cite{GaJoSt76} which allows, for any graph $G$, to produce a planar graph $G'$ with $O(n)$ vertices which is equivalent to $G$ for the \textsc{Maximum Independent set} problem. Concretely, $G'$ is obtained by replacing each crossing appearing in an embedding of $G$ in the plane by this gadget. In this article, we call it the \textit{GJS-gadget} (for Garey-Johnson-Stockmeyer) and denote it by $\gjs$ (see Fig.~\ref{subfig:gadget_embed}).
Unfortunately, this trick does not work directly for \textsc{$1$-Extendability}. In order to make it work, our idea is to define a first reduction producing a non-planar graph, but where the crossings satisfy some interesting properties. Secondly, we use the previously mentioned gadget on this intermediate graph.
Lastly, we use well-known tricks from the literature in order to reduce the maximum degree of the reduced graph, and to obtain a unit disk graph.

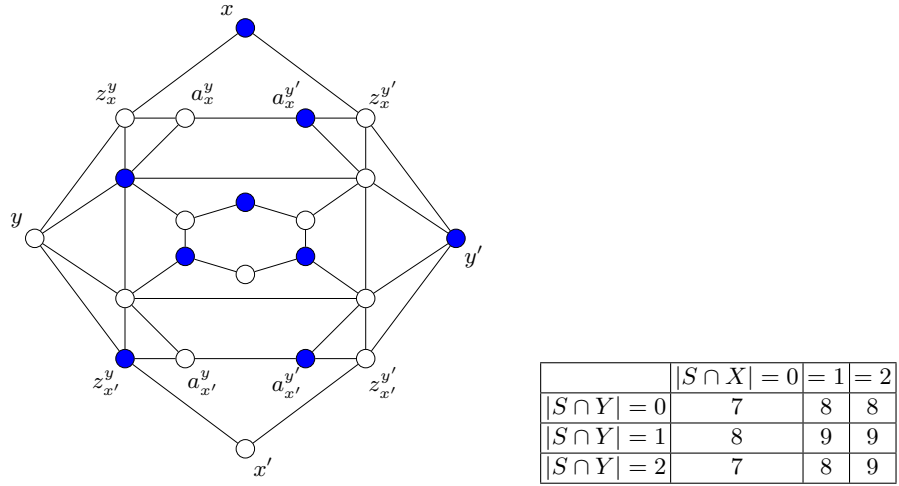
\begin{figure}[h]
\begin{subfigure}[b]{0.59\columnwidth}
\centering
\scalebox{0.8}{\begin{tikzpicture}[rotate=-90]
	\node[draw,circle,fill = blue] (x) at (0.5,4.0) {};
		\node[scale = 1.2] (xx) at (0.2,3.7) {$x$};
	\node[draw,circle] (x') at (7.5,4.0) {};
		\node[scale = 1.2] (xx') at (7.8,4.3) {$x'$};
	\node[draw,circle] (y) at (4.0,0.5) {};
		\node[scale = 1.2] (yy) at (3.7,0.2) {$y$};
	\node[draw,circle, fill = blue] (y') at (4.0,7.5) {};
		\node[scale = 1.2] (yy') at (4.3,7.8) {$y'$};
	
	\node[draw,circle] (l1) at (2.0,2.0) {};
		\node[scale = 1.2] (xy) at (1.6,1.7) {$z_x^{y}$};
	\node[draw,circle] (l2) at (2.0,3.0) {};
		\node[scale = 1.2] (axy) at (1.6,3.3) {$a_x^{y}$};
	\node[draw,circle, fill = blue] (l3) at (2.0,5.0) {};
		\node[scale = 1.2] (axy') at (1.6,4.7) {$a_x^{y'}$};
	\node[draw,circle] (l4) at (2.0,6.0) {};
		\node[scale = 1.2] (xy') at (1.6,6.3) {$z_x^{y'}$};
	\node[draw,circle] (t1) at (3.0,6.0) {};
	\node[draw,circle] (t2) at (5.0,6.0) {};
	\node[draw,circle, fill = blue] (r1) at (6.0,2.0) {};
		\node[scale = 1.2] (x'y) at (6.4,1.7) {$z_{x'}^{y}$};
	\node[draw,circle] (r2) at (6.0,3.0) {};
		\node[scale = 1.2] (ax'y) at (6.4,3.3) {$a_{x'}^{y}$};
	\node[draw,circle, fill = blue] (r3) at (6.0,5.0) {};
		\node[scale = 1.2] (ax'y') at (6.4,4.7) {$a_{x'}^{y'}$};
	\node[draw,circle] (r4) at (6.0,6.0) {};
		\node[scale = 1.2] (x'y') at (6.4,6.3) {$z_{x'}^{y'}$};
	\node[draw,circle, fill = blue] (b1) at (3.0,2.0) {};
	\node[draw,circle] (b2) at (5.0,2.0) {};
	
	\node[draw,circle] (c1) at (3.7,3.0) {};
	\node[draw,circle, fill = blue] (c2) at (4.3,3.0) {};
	\node[draw,circle] (c3) at (4.6,4.0) {};
	\node[draw,circle] (c5) at (3.7,5.0) {};
	\node[draw,circle, fill = blue] (c4) at (4.3,5.0) {};
	\node[draw,circle, fill = blue] (c6) at (3.4,4.0) {};
	
	\draw (x)--(l1);
	\draw (x)--(l4);
	
	\draw (l1)--(l2);
	\draw (l2)--(l3);
	\draw (l3)--(l4);
	
	\draw (x')--(r1);
	\draw (x')--(r4);
	
	\draw (r1)--(r2);
	\draw (r2)--(r3);
	\draw (r3)--(r4);
	
	\draw (y')--(l4);
	\draw (y')--(t1);
	\draw (y')--(t2);
	\draw (y')--(r4);
	
	\draw (l4)--(t1);
	\draw (t1)--(t2);
	\draw (t2)--(r4);
	
	\draw (l1)--(b1);
	\draw (b1)--(b2);
	\draw (b2)--(r1);
	
	\draw (y)--(l1);
	\draw (y)--(b1);
	\draw (y)--(b2);
	\draw (y)--(r1);
	
	\draw (l3)--(t1);
	\draw (r3)--(t2);
	\draw (l2)--(b1);
	\draw (r2)--(b2);
	\draw (b1)--(t1);
	\draw (b2)--(t2);
	
	\draw (c1)--(c2);
	\draw (c2)--(c3);
	\draw (c3)--(c4);
	\draw (c4)--(c5);
	\draw (c5)--(c6);
	\draw (c6)--(c1);
	
	\draw (b1)--(c1);
	\draw (b2)--(c2);
	\draw (t2)--(c4);
	\draw (t1)--(c5);
\end{tikzpicture}}
\caption{A planar embedding of $\gjs$ together with an MIS of it (in blue)}
\label{subfig:gadget_embed}
\end{subfigure}
~
\begin{subfigure}[b]{0.39\columnwidth}
\centering
$\begin{array}{|c|c|c|c|}
\hline
 & \card{S \cap X} = 0  & = 1 & = 2\\
 \hline
 \card{S \cap Y} = 0 & 7 & 8 & 8\\
 \hline
 \card{S \cap Y} = 1 & 8 & 9 & 9\\
 \hline
 \card{S \cap Y} = 2 & 7 & 8 & 9\\
 \hline
\end{array}$
\caption{Largest MISs $S$ of $\gjs$ containing a certain subset}
\label{subfig:gadget_sizes}
\end{subfigure}
\caption{The GJS-gadget~\cite{GaJoSt76}}
\label{fig:gadget}
\end{figure}

\textbf{Description of the gadget}. Fig.~\ref{subfig:gadget_embed} represents $\gjs$. Let $X = \set{x,x'}$, $Y = \set{y,y'}$, $Z = \set{z_{x}^{y},z_{x}^{y'},z_{x'}^{y},z_{x'}^{y'}}$, $A = \set{a_{x}^{y},a_{x}^{y'},a_{x'}^{y},a_{x'}^{y'}}$. We denote by $b_{x}^{y}$ the common neighbor of $z_x^y$ and $a_x^y$. Vertices $b_{x}^{y'}, b_{x'}^{y}, b_{x'}^{y'}$ are defined similarly. We fix $B = \set{b_{x}^{y},b_{x}^{y'},b_{x'}^{y},b_{x'}^{y'}}$. The ``$C_6$'' of $\gjs$ refers to the vertices which are not in sets $X,Y,Z,A,$ and $B$. The size of the MIS of $\gjs$ is 9, according to~\cite{GaJoSt76}. Blue vertices give an example of such MIS.
Vertices $x$, $x'$, $y$ and $y'$ are called the \textit{endpoints} of $\gjs$.

Fig.~\ref{subfig:gadget_sizes} indicates the sizes of a largest independent set $S$ we obtain if we fix the intersection size with $X$ and $Y$. For example, a largest independent set $S$ which contains vertices $x,y,y'$ is of size 8: one of them is such that it also contains $a_x^y,a_{x'}^y$ and 3 vertices from the $C_6$. Another example: a largest independent set $S$ containing exactly one vertex of $X$ and one vertex of $Y$ has size 9. The blue vertices of Fig.~\ref{subfig:gadget_embed} form this kind of independent sets, with $S \cap X = \set{x}$ and $S \cap Y = \set{y'}$.

Consider an embedding of some graph $G$ in the plane, 
and a crossing consisting of two edges $uu'$ and $vv'$ (as, for instance, in Fig.~\ref{subfig:crossings}, where $v_1v_1'$ plays the role of $vv'$). By \textit{replacing the crossing by a gadget}, we mean removing the edges $uu'$ and $vv'$, adding a subgraph isomorphic to $\gjs$, and adding the edges $vx$, $uy'$, $v'x'$, and $u'y$. 
By replacing each crossing of $G$ by a gadget, we obtain a graph $G_+$ which is not only planar, but also equivalent to $G$ for the \textsc{Maximum Independent Set} problem, in the sense that $G$ contains an independent set of size $k$ iff $G_+$ contains an independent set of size $k + 9\lambda$, where $\lambda$ is the number of crossings in $G$~\cite{GaJoSt76}. Fig.~\ref{subfig:crossings} shows an example of edge $uu'$ of some graph $G$ which is crossed by three other edges $v_1v_1'$, $v_2v_2'$, and $v_3v_3'$. In $G_+$, these crossings become graphs isomorphic to $\gjs$: they are denoted by $H_1$, $H_2$, and $H_3$ respectively (Fig.~\ref{subfig:crossings_gadget}). Observe that, in Fig.~\ref{subfig:crossings_gadget}, gadgets $H_1$, $H_2$, and $H_3$ are oriented in the same way (to avoid confusions), but it is not a mandatory requirement for the reduction. In other words, turning $H_1$ so that the neighbor of $u$ (resp. the neighbors of $v_1,v_1'$) has degree four (resp. two) still would work.

\begin{figure}[h]
\begin{subfigure}[b]{0.32\columnwidth}
\centering
\scalebox{0.8}{\begin{tikzpicture}
	\node[draw,circle] (u) at (4.0,9.0) {};
	\node[draw,circle] (u') at (4.0,1.0) {};
	\node[draw,circle] (v3) at (2.5,3.0) {};
	\node[draw,circle] (v3') at (5.5,3.0) {};
	\node[draw,circle] (v2) at (2.5,5.0) {};
	\node[draw,circle] (v2') at (5.5,5.0) {};
	\node[draw,circle] (v1) at (2.5,7.0) {};
	\node[draw,circle] (v1') at (5.5,7.0) {};
	
	\draw (u)--(u');
	\draw (v1)--(v1');
	\draw (v2)--(v2');
	\draw (v3)--(v3');
	
	\node[scale = 1.2] (ux) at (3.6,9.0) {$u$};
	\node[scale = 1.2] (u'x) at (4.4,1.0) {$u'$};
	\node[scale = 1.2] (v1x) at (2.1,7.0) {$v_1$};
	\node[scale = 1.2] (v1'x) at (5.9,7.0) {$v_1'$};
	\node[scale = 1.2] (v2x) at (2.1,5.0) {$v_2$};
	\node[scale = 1.2] (v2'x) at (5.9,5.0) {$v_2'$};
	\node[scale = 1.2] (v3x) at (2.1,3.0) {$v_3$};
	\node[scale = 1.2] (v3'x) at (5.9,3.0) {$v_3'$};
\end{tikzpicture}}
\caption{Edge $uu'$ crossing $v_1v_1',v_2v_2',v_3v_3'$}
\label{subfig:crossings}
\end{subfigure}
~
\begin{subfigure}[b]{0.32\columnwidth}
\centering
\scalebox{0.8}{\begin{tikzpicture}
	\node[draw,circle] (u) at (4.0,9.0) {};
	\node[draw,circle] (u') at (4.0,1.0) {};
	\node[draw,circle] (v3) at (2.5,3.0) {};
	\node[draw,circle] (v3') at (5.5,3.0) {};
	\node[draw,circle] (v2) at (2.5,5.0) {};
	\node[draw,circle] (v2') at (5.5,5.0) {};
	\node[draw,circle] (v1) at (2.5,7.0) {};
	\node[draw,circle] (v1') at (5.5,7.0) {};
	
	\node[draw,circle] (x_1) at (4.0,7.5) {};
	\node[draw,circle] (x_1') at (4.0,6.5) {};
	\node[draw,circle] (x_2) at (4.0,5.5) {};
	\node[draw,circle] (x_2') at (4.0,4.5) {};
	\node[draw,circle] (x_3) at (4.0,3.5) {};
	\node[draw,circle] (x_3') at (4.0,2.5) {};
	
	\node[draw,circle] (y_1) at (3.5,7.0) {};
	\node[draw,circle] (y_1') at (4.5,7.0) {};
	\node[draw,circle] (y_2) at (3.5,5.0) {};
	\node[draw,circle] (y_2') at (4.5,5.0) {};
	\node[draw,circle] (y_3) at (3.5,3.0) {};
	\node[draw,circle] (y_3') at (4.5,3.0) {};
	
	\draw[color = red, dashed] (3.3,7.7) -- (3.3,6.3) -- (4.7,6.3) -- (4.7,7.7) -- (3.3,7.7);
	\draw[color = red, dashed] (3.3,5.7) -- (3.3,4.3) -- (4.7,4.3) -- (4.7,5.7) -- (3.3,5.7);
	\draw[color = red, dashed] (3.3,3.7) -- (3.3,2.3) -- (4.7,2.3) -- (4.7,3.7) -- (3.3,3.7);
	
	\draw (u)--(x_1);
	\draw (x_1')--(x_2);
	\draw (x_2')--(x_3);
	\draw (x_3')--(u');
	
	\draw (v1)--(y_1);
	\draw (y_1')--(v1');
	
	\draw (v2)--(y_2);
	\draw (y_2')--(v2');
	
	\draw (v3)--(y_3);
	\draw (y_3')--(v3');
	
	\draw (x_1)--(3.8,7.3);
	\draw (x_1)--(4.2,7.3);
	
	\draw (y_1)--(3.8,7.1);
	\draw (y_1)--(3.8,6.9);
	\draw (y_1)--(3.7,7.2);
	\draw (y_1)--(3.7,6.8);
	
	\draw (x_1')--(3.8,6.7);
	\draw (x_1')--(4.2,6.7);
	
	\draw (y_1')--(4.2,7.1);
	\draw (y_1')--(4.2,6.9);
	\draw (y_1')--(4.3,7.2);
	\draw (y_1')--(4.3,6.8);

	\draw (x_2)--(3.8,5.3);
	\draw (x_2)--(4.2,5.3);
	
	\draw (y_2)--(3.8,5.1);
	\draw (y_2)--(3.8,4.9);
	\draw (y_2)--(3.7,5.2);
	\draw (y_2)--(3.7,4.8);
	
	\draw (x_2')--(3.8,4.7);
	\draw (x_2')--(4.2,4.7);
	
	\draw (y_2')--(4.2,5.1);
	\draw (y_2')--(4.2,4.9);
	\draw (y_2')--(4.3,5.2);
	\draw (y_2')--(4.3,4.8);

	\draw (x_3)--(3.8,3.3);
	\draw (x_3)--(4.2,3.3);
	
	\draw (y_3)--(3.8,3.1);
	\draw (y_3)--(3.8,2.9);
	\draw (y_3)--(3.7,3.2);
	\draw (y_3)--(3.7,2.8);
	
	\draw (x_3')--(3.8,2.7);
	\draw (x_3')--(4.2,2.7);
	
	\draw (y_3')--(4.2,3.1);
	\draw (y_3')--(4.2,2.9);
	\draw (y_3')--(4.3,3.2);
	\draw (y_3')--(4.3,2.8);
	
	\node[scale = 1.2] (ux) at (3.6,9.0) {$u$};
	\node[scale = 1.2] (u'x) at (4.4,1.0) {$u'$};
	\node[scale = 1.2] (v1x) at (2.1,7.0) {$v_1$};
	\node[scale = 1.2] (v1'x) at (5.9,7.0) {$v_1'$};
	\node[scale = 1.2] (v2x) at (2.1,5.0) {$v_2$};
	\node[scale = 1.2] (v2'x) at (5.9,5.0) {$v_2'$};
	\node[scale = 1.2] (v3x) at (2.1,3.0) {$v_3$};
	\node[scale = 1.2] (v3'x) at (5.9,3.0) {$v_3'$};
	
	\node[scale = 1.2, color = red] (h1) at (2.9,7.4) {$H_1$};
	\node[scale = 1.2, color = red] (h2) at (2.9,5.4) {$H_2$};
	\node[scale = 1.2, color = red] (h3) at (2.9,3.4) {$H_3$};
\end{tikzpicture}}
\caption{After replacing each crossing in $G_+$}
\label{subfig:crossings_gadget}
\end{subfigure}
~
\begin{subfigure}[b]{0.32\columnwidth}
\centering
\scalebox{0.8}{\begin{tikzpicture}
	\node[draw,circle,fill = black!20!blue] (u) at (4.0,9.0) {};
	\node[draw,circle] (u') at (4.0,1.0) {};
	\node[draw,circle] (v3) at (2.5,3.0) {};
	\node[draw,circle] (v3') at (5.5,3.0) {};
	\node[draw,circle,fill = black!20!blue] (v2) at (2.5,5.0) {};
	\node[draw,circle] (v2') at (5.5,5.0) {};
	\node[draw,circle] (v1) at (2.5,7.0) {};
	\node[draw,circle,fill = black!20!blue] (v1') at (5.5,7.0) {};
	
	\node[draw,circle] (x_1) at (4.0,7.5) {};
	\node[draw,circle,fill=white!50!blue] (x_1') at (4.0,6.5) {};
	\node[draw,circle] (x_2) at (4.0,5.5) {};
	\node[draw,circle,fill=white!50!blue] (x_2') at (4.0,4.5) {};
	\node[draw,circle] (x_3) at (4.0,3.5) {};
	\node[draw,circle,fill=white!50!blue] (x_3') at (4.0,2.5) {};
	
	\node[draw,circle,fill=white!50!blue] (y_1) at (3.5,7.0) {};
	\node[draw,circle] (y_1') at (4.5,7.0) {};
	\node[draw,circle] (y_2) at (3.5,5.0) {};
	\node[draw,circle,fill=white!50!blue] (y_2') at (4.5,5.0) {};
	\node[draw,circle] (y_3) at (3.5,3.0) {};
	\node[draw,circle,fill=white!50!blue] (y_3') at (4.5,3.0) {};
	
	\draw[color = red, dashed] (3.3,7.7) -- (3.3,6.3) -- (4.7,6.3) -- (4.7,7.7) -- (3.3,7.7);
	\draw[color = red, dashed] (3.3,5.7) -- (3.3,4.3) -- (4.7,4.3) -- (4.7,5.7) -- (3.3,5.7);
	\draw[color = red, dashed] (3.3,3.7) -- (3.3,2.3) -- (4.7,2.3) -- (4.7,3.7) -- (3.3,3.7);
	
	\draw (u)--(x_1);
	\draw (x_1')--(x_2);
	\draw (x_2')--(x_3);
	\draw (x_3')--(u');
	
	\draw (v1)--(y_1);
	\draw (y_1')--(v1');
	
	\draw (v2)--(y_2);
	\draw (y_2')--(v2');
	
	\draw (v3)--(y_3);
	\draw (y_3')--(v3');
	
	\draw (x_1)--(3.8,7.3);
	\draw (x_1)--(4.2,7.3);
	
	\draw (y_1)--(3.8,7.1);
	\draw (y_1)--(3.8,6.9);
	\draw (y_1)--(3.7,7.2);
	\draw (y_1)--(3.7,6.8);
	
	\draw (x_1')--(3.8,6.7);
	\draw (x_1')--(4.2,6.7);
	
	\draw (y_1')--(4.2,7.1);
	\draw (y_1')--(4.2,6.9);
	\draw (y_1')--(4.3,7.2);
	\draw (y_1')--(4.3,6.8);

	\draw (x_2)--(3.8,5.3);
	\draw (x_2)--(4.2,5.3);
	
	\draw (y_2)--(3.8,5.1);
	\draw (y_2)--(3.8,4.9);
	\draw (y_2)--(3.7,5.2);
	\draw (y_2)--(3.7,4.8);
	
	\draw (x_2')--(3.8,4.7);
	\draw (x_2')--(4.2,4.7);
	
	\draw (y_2')--(4.2,5.1);
	\draw (y_2')--(4.2,4.9);
	\draw (y_2')--(4.3,5.2);
	\draw (y_2')--(4.3,4.8);

	\draw (x_3)--(3.8,3.3);
	\draw (x_3)--(4.2,3.3);
	
	\draw (y_3)--(3.8,3.1);
	\draw (y_3)--(3.8,2.9);
	\draw (y_3)--(3.7,3.2);
	\draw (y_3)--(3.7,2.8);
	
	\draw (x_3')--(3.8,2.7);
	\draw (x_3')--(4.2,2.7);
	
	\draw (y_3')--(4.2,3.1);
	\draw (y_3')--(4.2,2.9);
	\draw (y_3')--(4.3,3.2);
	\draw (y_3')--(4.3,2.8);
	
	\node[scale = 1.2] (ux) at (3.6,9.0) {$u$};
	\node[scale = 1.2] (u'x) at (4.4,1.0) {$u'$};
	\node[scale = 1.2] (v1x) at (2.1,7.0) {$v_1$};
	\node[scale = 1.2] (v1'x) at (5.9,7.0) {$v_1'$};
	\node[scale = 1.2] (v2x) at (2.1,5.0) {$v_2$};
	\node[scale = 1.2] (v2'x) at (5.9,5.0) {$v_2'$};
	\node[scale = 1.2] (v3x) at (2.1,3.0) {$v_3$};
	\node[scale = 1.2] (v3'x) at (5.9,3.0) {$v_3'$};
	
	\node[scale = 1.2, color = red] (h1) at (2.9,7.4) {$H_1$};
	\node[scale = 1.2, color = red] (h2) at (2.9,5.4) {$H_2$};
	\node[scale = 1.2, color = red] (h3) at (2.9,3.4) {$H_3$};
\end{tikzpicture}}
\caption{An of $G$ completed with 9 vertices per crossing}
\label{subfig:crossings_mis}
\end{subfigure}
\caption{Replacing each crossing by a GJS-gadget}
\label{fig:crossings}
\end{figure}
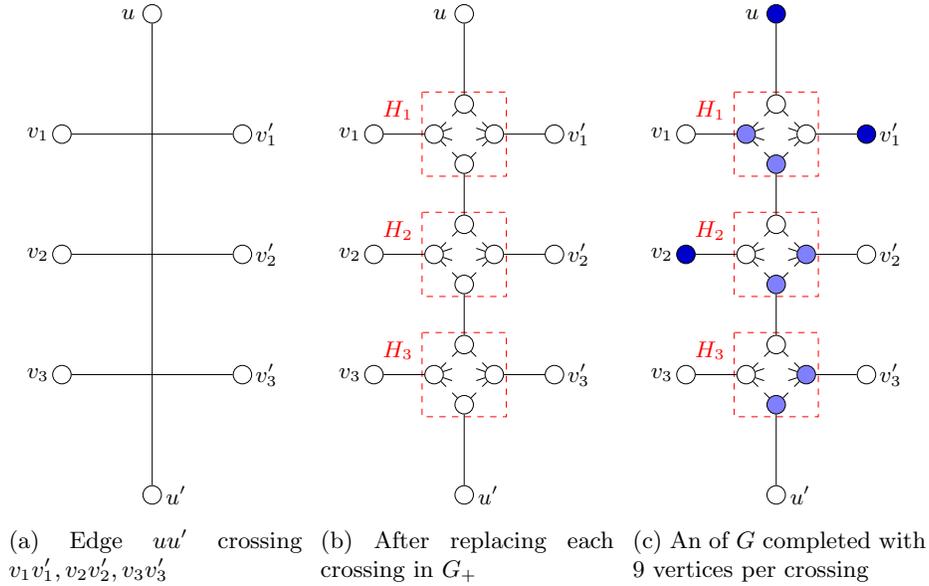

As said previously, the size of an MIS $G_+$ is $\alpha(G) + 9\lambda$, where $\lambda$ is the number of crossings in the planar embedding of $G$. The idea behind this statement is the following: if $S$ is an independent set of $G$, then there exists an independent set $S_+$ of $G_+$ of size $\card{S} + 9\lambda$ which is made up of the vertices of $S$ and 9 vertices per crossing. In Fig.~\ref{subfig:crossings_mis}, the dark blue vertices represent an independent set $S$ of $G$ and the light blue ones are the endpoints of the gadgets, {\em i.e.} $x,x',y,y'$ in Fig.~\ref{fig:gadget}, which belong to $S_+$ in $G_+$. As $S$ is independent, one can select, for each gadget, one vertex of $\set{x,x'}$ (and one vertex of $\set{y,y'}$) which is not adjacent to an element of $S$. We know that a largest independent set of $\gjs$ intersecting both $\set{x,x'}$ and $\set{y,y'}$ in exactly one element has size 9, which corresponds to the MIS size of $\gjs$.

\begin{lemma}[Graphs $G$ and $G_+$ are equivalent for \textsc{Maximum Independent Set}~\cite{GaJoSt76}]
Any MIS $S$ of $G$ can be completed into an MIS $S_+ \supseteq S$ of $G_+$ which contains exactly 9 vertices per crossing gadget. Conversely, given any MIS $S^*$ of $G_+$, the vertices of $S^*$ which do not belong to a crossing gadget of $G_+$ form an MIS of $G$.
\label{le:eq_mis}
\end{lemma}

\textbf{Preservation of $1$-extendability}. Our initial idea was to use the same gadget to transform every graph into a planar one which preserves the $1$-extendability of $G$. Unfortunately, the property described above for \textsc{Maximum Independent Set} does not hold for \textsc{$1$-Extendability}. Indeed, one can find examples of graphs $G$ such that $G$ is 1-extendable and $G_+$ is not. For this reason, we state a weaker characterization involving the GJS-gadget. We will see further that this result is enough to prove that \textsc{$1$-Extendability} is NP-hard on planar graphs.

\begin{proposition}
Let $G$ be a graph embedded in the plane and $uu' \in E(G)$. Let $v_1v_1',v_2v_2',\ldots,v_{\ell}v_{\ell}'$ be the edges of $G$ which cross $uu'$. Assume, for any $1\le i\le \ell$, the following statements:
\begin{itemize}
\item there is an MIS $S_u^{(i)}$ of $G$ such that $S_u^{(i)} \cap \set{u,u',v_i,v_i'} = \set{u}$,
\item there is an MIS $S_{u'}^{(i)}$ of $G$ such that $S_{u'}^{(i)} \cap \set{u,u',v_i,v_i'} = \set{u'}$,
\end{itemize}
Let $G_+$ be the graph obtained from $G$ by replacing each crossing $\set{uu',v_iv_i'}$ with a GJS-gadget. Then, $G_+$ is 1-extendable iff $G$ is 1-extendable.
\label{prop:garey-johnson}
\end{proposition}
\begin{proof}
One direction is trivial. If $G_+$ is 1-extendable, then for any vertex $u$ outside a crossing gadget, {\em i.e.} $u \in V(G)$, there is an MIS $S^*$ of $G_+$ containing $u$. According to Lemma~\ref{le:eq_mis}, one can obtain an MIS $S$ of $G$ containing $u$ by removing the vertices of $S$ belonging to the gadgets. In brief, for any $u \in V(G)$, there is an MIS of $G$ containing $u$.

We suppose now that $G$ is 1-extendable. This direction is trickier. We begin with a few notation. Let $H_i$ be the crossing gadget $\gjs$ of $\set{uu',v_iv_i'}$ (see Fig.~\ref{subfig:crossings_gadget}), and $x_i,x_i',y_i,y_i'$ denote the endpoints of $H_i$, $1\le i\le \ell$.
 We fix $x_i$ as the closest one to $u$ and $y_i$ as the closest one to $v_i$. Our objective is to prove that each vertex of gadgets $H_i$ belongs to some MIS of $G_+$. Indeed, we already know that the vertices of $V(G) \subseteq V(G_+)$ are contained in an MIS of $G_+$, according to Lemma~\ref{le:eq_mis}.

We pursue with an observation on the gadget $\gjs$. Let $a \in \{x, x'\}$ and $b \in \{y, y'\}$. Observe that there is a single MIS of $\gjs$ whose intersection with $\{x, x'\}$ is $\{a\}$ and whose intersection with $\{y, y'\}$ is $\{b\}$. We denote by $S_{ab}$ this set (for instance, the set $S_{xy'}$ is depicted in Fig.~\ref{subfig:gadget_embed}). In addition, we have $S_{xy} \cup S_{xy'} \cup S_{x'y} \cup S_{x'y'} = V(\gjs)$. Consequently, in order to prove that $G_+$ is 1-extendable, it is sufficient to show that there are MISs of $G_+$ intersecting $\set{x_i,x_i',y_i,y_i'}$ exactly in $\set{x_i,y_i}$, $\set{x_i,y_i'}$, $\set{x_i',y_i}$, and $\set{x_i',y_i'}$ respectively.
The remainder consists in using both Lemma~\ref{le:eq_mis} and the assumptions on the MISs of $G$ to put in evidence MISs of $G_+$ which intersect exactly these pairs.
Let $S_u^{(i)}$ and $S_{u'}^{(i)}$ be as in the statement. We now show how to complete them in order to obtain these MISs.

\textit{Completion of $S_u^{(i)}$}. According to Lemma~\ref{le:eq_mis}, one can produce an MIS $S_{u,+}^{(i)}$ in $G_+$ which contains exactly 9 vertices per crossing gadget and $S_u^{(i)} \subseteq S_{u,+}^{(i)}$. All vertices $x_j'$, $1\le j\le \ell$ necessarily belong to $S_{u,+}^{(i)}$. Indeed, $x_1$ is adjacent to $u$, so if we aim at picking up 9 vertices in $H_1$, according to Fig.~\ref{subfig:gadget_sizes}, $x_1' \in S_{u,+}^{(i)}$. Then, $x_2$ is adjacent to $x_1'$, so $x_2'$ must be picked up, etc. In particular, $x_i' \in S_{u,+}^{(i)}$. Then, we know that neither $v_i$ nor $v_i'$ are in $S_u^{(i)}$. Therefore, $S_{u,+}^{(i)}$ may contain either $y_i$ or $y_i'$ or both of them. But, it suffices to pick up exactly one of them to have $\card{H_i \cap S_{u,+}^{(i)}} = 9$. Moreover, selecting either $y_i$ or $y_i'$ when we produce $S_{u,+}^{(i)}$ does not influence the adjacency over the other crossing gadgets or the rest of the graph. Hence, there are two MISs of $G_+$: one intersecting exactly $\set{x_i',y_i}$ and another intersecting exactly $\set{x_i',y_i'}$.

\textit{Completion of $S_{u'}^{(i)}$}. The symmetrical analysis provides us with two MISs of $G_+$ intersecting the endpoints of $H_i$ on exactly $\set{x_i,y_i}$ and $\set{x_i,y_i'}$ respectively.

In summary, for any $1\le i\le \ell$, there are MISs which respectively intersect the set $\set{x_i,x_i',y_i,y_i'}$ in pairs $\set{x_i,y_i}$, $\set{x_i,y_i'}$, $\set{x_i',y_i}$, and $\set{x_i',y_i'}$. Referring to our previous observation, this ensures us that all vertices in gadgets $H_i$ are covered by MISs of $G_+$. As a conclusion, $G_+$ is 1-extendable.
\end{proof}

Observe that the assumptions concerning the MISs of $G_+$ are essential if we want pairs $\set{x_i,y_i}$, $\set{x_i,y_i'}$, $\set{x_i',y_i}$, and $\set{x_i',y_i'}$ of each gadget $H_i$ to be covered by MISs. This property is not achieved by all 1-extendable graphs $G$: take for instance an embedding of some complete bipartite graph $K_{n,n}$ with $n\ge 3$, every MIS intersects each crossing on exactly two vertices.

\subsection{Planar embedding}

The GJS-gadget is a key tool in our proof that \textsc{$1$-Extendability} is NP-hard on planar graphs. We reduce from an NP-hard variant of \textsc{3SAT} called \textsc{Planar Monotone Rectilinear 3SAT}, abbreviated \textsc{PMR 3SAT}. Given an input $\varphi$ of \textsc{PMR 3SAT}, we design a graph $G_{\varphi}$ such that $\varphi$ is satisfiable iff $G_{\varphi}$ is 1-extendable. Furthermore, $G_{\varphi}$ is planar and its maximum degree is 3. We begin with the construction of $G_{\varphi}$ step by step. Then, we show that the $1$-extendability of $G_{\varphi}$ depends on the satisfiability of the formula $\varphi$.

\textbf{Starting point of the reduction}. We reduce from \textsc{PMR 3SAT}, which is NP-hard~\cite{BeKh10}. In this variant of \textsc{3SAT}, clauses and variables can be represented in the plane in a certain way. The input is a set of variables $X = \set{x_1,\ldots,x_n}$ and a CNF-SAT formula $\varphi$ over $X$ with exactly three variables per clause. The clauses $C_1,\ldots,C_m$ are \textit{monotone}: they contain either three positive literals or three negative literals. Moreover, $\varphi$ admits a \textit{rectilinear} representation, that we now explain. Each variable is a point on the x-axis. The positive (resp. negative) clauses are represented by horizontal segments above (resp. below) the x-axis. When a variable $x_i$ appears in a given clause, a vertical edge must connect the point $x_i$ on the x-axis with the segment of this clause (at any point of the segment). Such a representation is rectilinear if no edge crosses a clause segment. Fig.~\ref{fig:pmr3sat} provides an example of a formula $\varphi$, with $m=5$, which admits a rectilinear representation.

\begin{figure}[h]
\centering
\scalebox{0.8}{\begin{tikzpicture}


\draw[->, >=latex, color = black, line width = 1.4pt] (0.0,8.0) -- (12.0,8.0);
\node[scale = 1.2] at (12.2,7.7) {$\mbox{x}$};

\draw [color = black, fill = white] (2.0,7.7) -- (3.0,7.7) -- (3.0,8.3) -- (2.0,8.3) --  (2.0,7.7);
\node[scale = 1.1] at (2.5,8.0) {$x_1$};
\draw [color = black, fill = white] (4.0,7.7) -- (5.0,7.7) -- (5.0,8.3) -- (4.0,8.3) --  (4.0,7.7);
\node[scale = 1.1] at (4.5,8.0) {$x_2$};
\draw [color = black, fill = white] (6.0,7.7) -- (7.0,7.7) -- (7.0,8.3) -- (6.0,8.3) --  (6.0,7.7);
\node[scale = 1.1] at (6.5,8.0) {$x_3$};
\draw [color = black, fill = white] (8.0,7.7) -- (9.0,7.7) -- (9.0,8.3) -- (8.0,8.3) --  (8.0,7.7);
\node[scale = 1.1] at (8.5,8.0) {$x_4$};
\draw [color = black, fill = white] (10.0,7.7) -- (11.0,7.7) -- (11.0,8.3) -- (10.0,8.3) --  (10.0,7.7);
\node[scale = 1.1] at (10.5,8.0) {$x_5$};


\draw [color = black] (4.5,9.5) -- (8.5,9.5) -- (8.5,10.0) -- (4.5,10.0) --  (4.5,9.5);
\node at (6.5,9.75) {$x_2 \vee x_3 \vee x_4$};
\draw [color = black] (2.5,10.5) -- (10.5,10.5) -- (10.5,11.0) -- (2.5,11.0) --  (2.5,10.5);
\node at (6.5,10.75) {$x_1 \vee x_2 \vee x_5$};

\draw [color = black] (2.5,6.5) -- (6.4,6.5) -- (6.4,6.0) -- (2.5,6.0) --  (2.5,6.5);
\node at (4.5,6.25) {$\neg x_1 \vee \neg x_2 \vee \neg x_3$};
\draw [color = black] (6.6,6.5) -- (10.5,6.5) -- (10.5,6.0) -- (6.6,6.0) --  (6.6,6.5);
\node at (8.5,6.25) {$\neg x_3 \vee \neg x_4 \vee \neg x_5$};
\draw [color = black] (2.4,5.5) -- (10.6,5.5) -- (10.6,5.0) -- (2.4,5.0) --  (2.4,5.5);
\node at (6.5,5.25) {$\neg x_1 \vee \neg x_3 \vee \neg x_5$};


\draw[color=black] (2.7,8.3) -- (2.7,10.5);
\draw[color=black] (4.3,8.3) -- (4.3,10.5);
\draw[color=black] (4.7,8.3) -- (4.7,9.5);
\draw[color=black] (6.5,8.3) -- (6.5,9.5);
\draw[color=black] (8.3,8.3) -- (8.3,9.5);
\draw[color=black] (10.3,8.3) -- (10.3,10.5);

\draw[color=black] (2.7,7.7) -- (2.7,6.5);
\draw[color=black] (4.5,7.7) -- (4.5,6.5);
\draw[color=black] (6.3,7.7) -- (6.3,6.5);
\draw[color=black] (6.7,7.7) -- (6.7,6.5);
\draw[color=black] (8.5,7.7) -- (8.5,6.5);
\draw[color=black] (10.3,7.7) -- (10.3,6.5);

\draw[color=black] (2.4,7.7) -- (2.4,5.5);
\draw[color=black] (6.5,7.7) -- (6.5,5.5);
\draw[color=black] (10.6,7.7) -- (10.6,5.5);

\end{tikzpicture}}
\caption{A rectilinear representation of a PMR 3SAT instance $C_1,\ldots,C_5$}
\label{fig:pmr3sat}
\end{figure}
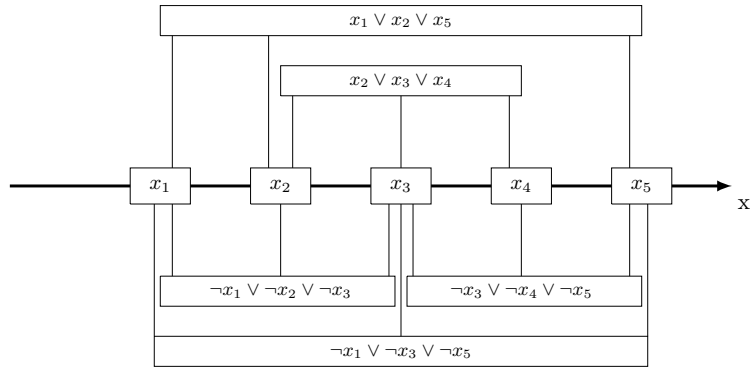

Let $\varphi$ be an input of PMR 3SAT provided with its rectilinear representation. The construction of $G_{\varphi}$ depends on the rectilinear representation of $\varphi$. We proceed with two intermediate steps: first graph $G_{\varphi}''$, second graph $G_{\varphi}'$.

\textbf{Construction of} $G_{\varphi}''$. The first step is inspired from Mohar's reduction~\cite{Mo01} for \textsc{Maximum Independent Set}. We replace each variable $x_i$ on the x-axis by a cycle. Let $r$ be the number of appearances of $x_i$ (as a literal $x_i$ or $\neg x_i$) in the clauses $C_1,\ldots,C_m$ of  $\varphi$. The point representing variable $x_i$ becomes a cycle $x_i^1,\bar{x}_i^1,x_i^2,\bar{x}_i^2,\ldots,x_i^r,\bar{x}_i^r$ of length $2r$, drawn as an axis-parallel rectangle (see Fig.~\ref{subfig:first_step}). 
We denote by $c^*$ the total number of vertices in the variable cycles. Each clause $C_j = \ell_j^1 \vee \ell_j^2 \vee \ell_j^3$ is replaced by a triangle $T_j$ of three vertices $v_j^1,v_j^2,v_j^3$. The edges of these triangles are called $T$-\textit{edges}. Each vertex of the clause is placed at the intersection between the clause segment and vertical edges of the rectilinear representation. In this way, vertices $v_j^1$, $v_j^2$ and $v_j^3$ are aligned horizontally and, w.l.o.g, we assume $v_j^1$ (resp. $v_j^3$) is the leftmost (resp. rightmost) vertex of $T_j$ on the clause segment. Edges $v_j^1v_j^2$ and $v_j^2v_j^3$ are drawn as straight lines. The third one, $v_j^1v_j^3$, can be represented as an almost flat curve, passing above (resp. below) vertex $v_2^j$ for positive (resp. negative) clauses. If $\ell_j^q = x_i$ for some $1\le j\le m$ and $q\in \set{1,2,3}$, then vertex $v_j^q$ is connected to some cycle vertex $\bar{x}_i^s$ of the top of the rectangle. Otherwise, if $\ell_j^q = \neg x_i$, then vertex $v_j^q$ is connected to some cycle vertex $x_i^s$ of the bottom of the rectangle. For now, the described embedding is planar. Fig.~\ref{subfig:first_step} shows the embedding of the instance of Fig.~\ref{fig:pmr3sat}. Vertices $\bar{x}_i^s$ are drawn in grey to distinguish them from vertices $x_i^s$ (in white).

Less formally, each parity of a variable cycle represents a certain assignation of this variable. Picking up $x_i^1,x_i^2,\ldots$ (resp. $\bar{x}_i^1,\bar{x}_i^2,\ldots$) into an independent set will correspond to assigning $x_i$ to \textsf{False} (resp. \textsf{True}). 

We add a ``pendant'' vertex $\pi_j$ for any triangle $T_j$, $1\le j\le m$, that is, $\pi_j$ is adjacent to all vertices of $T_j$. The edges created by this operation, {\em i.e.} all $v_j^q\pi_j$, are called \textit{pendant edges}. Consider the following embedding. We fix two horizontal axes $\mbox{x}^+$ and $\mbox{x}^-$: the first one above the x-axis and all segments of the positive clauses, the second one below the x-axis and all segments of the negative clauses. The pendants issued from the positive clauses are placed on the $\mbox{x}^+$-axis such that every edge $(v_j^2,\pi_j)$ is vertical. We represent edges $(v_j^1,\pi_j)$ and $(v_j^3,\pi_j)$ as straight lines (they cannot be vertical). We proceed similarly with pendants of the negative clauses on the $\mbox{x}^-$-axis. We denote by $G_{\varphi}''$ the obtained graph. Its embedding is not planar. Fig.~\ref{subfig:first_step_pendant} shows graph $G_{\varphi}''$ corresponding to the instance $\varphi$ of Fig.~\ref{fig:pmr3sat}. Pendant edges are drawn in red. We claim that each vertex of $G_{\varphi}''$ belongs to an MIS. This might seem counter-intuitive, but the equivalence between the $1$-extendability of the output instance and the satisfaction of $\varphi$ will appear later (when we will eventually define $G_{\varphi}$).

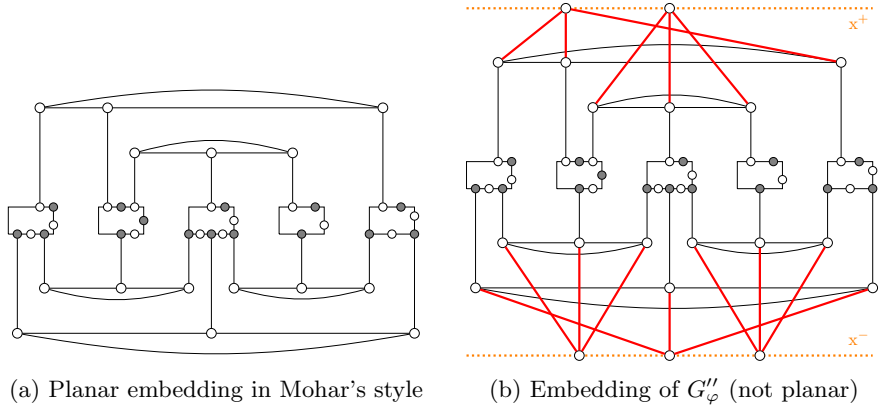
\begin{figure}[t]
\begin{subfigure}[b]{0.47\columnwidth}
\centering
\scalebox{0.6}{\begin{tikzpicture}



\draw [color = black] (2.0,7.7) -- (3.0,7.7) -- (3.0,8.3) -- (2.0,8.3) --  (2.0,7.7);
\draw [color = black] (4.0,7.7) -- (5.0,7.7) -- (5.0,8.3) -- (4.0,8.3) --  (4.0,7.7);
\draw [color = black] (6.0,7.7) -- (7.0,7.7) -- (7.0,8.3) -- (6.0,8.3) --  (6.0,7.7);
\draw [color = black] (8.0,7.7) -- (9.0,7.7) -- (9.0,8.3) -- (8.0,8.3) --  (8.0,7.7);
\draw [color = black] (10.0,7.7) -- (11.0,7.7) -- (11.0,8.3) -- (10.0,8.3) --  (10.0,7.7);


\node[draw,circle,scale = 0.6, fill = white] (x11) at (2.7,8.3) {};
\node[draw,circle,scale = 0.6, fill = white!50!black] (bx12) at (2.8,7.7) {};
\node[draw,circle,scale = 0.6, fill = white!50!black] (bx13) at (2.2,7.7) {};

\node[draw,circle,scale = 0.6, fill = white] (x21) at (4.2,8.3) {};
\node[draw,circle,scale = 0.6, fill = white] (x22) at (4.8,8.3) {};
\node[draw,circle,scale = 0.6, fill = white!50!black] (bx23) at (4.5,7.7) {};

\node[draw,circle,scale = 0.6, fill = white] (x31) at (6.5,8.3) {};
\node[draw,circle,scale = 0.6, fill = white!50!black] (bx32) at (7.0,7.7) {};
\node[draw,circle,scale = 0.6, fill = white!50!black] (bx33) at (6.5,7.7) {};
\node[draw,circle,scale = 0.6, fill = white!50!black] (bx34) at (6.0,7.7) {};

\node[draw,circle,scale = 0.6, fill = white] (x41) at (8.3,8.3) {};
\node[draw,circle,scale = 0.6, fill = white!50!black] (bx42) at (8.5,7.7) {};

\node[draw,circle,scale = 0.6, fill = white] (x51) at (10.3,8.3) {};
\node[draw,circle,scale = 0.6, fill = white!50!black] (bx52) at (11.0,7.7) {};
\node[draw,circle,scale = 0.6, fill = white!50!black] (bx53) at (10.0,7.7) {};


\node[draw,circle,scale = 0.6, fill = white!50!black] (bx11) at (3.0,8.3) {};
\node[draw,circle,scale = 0.6, fill = white] (x12) at (3.0,7.9) {};
\node[draw,circle,scale = 0.6, fill = white] (x13) at (2.5,7.7) {};

\node[draw,circle,scale = 0.6, fill = white!50!black] (bx21) at (4.5,8.3) {};
\node[draw,circle,scale = 0.6, fill = white!50!black] (bx22) at (5.0,8.0) {};
\node[draw,circle,scale = 0.6, fill = white] (x23) at (4.8,7.7) {};

\node[draw,circle,scale = 0.6, fill = white!50!black] (bx31) at (6.8,8.3) {};
\node[draw,circle,scale = 0.6, fill = white] (x32) at (7.0,8.0) {};
\node[draw,circle,scale = 0.6, fill = white] (x33) at (6.75,7.7) {};
\node[draw,circle,scale = 0.6, fill = white] (x34) at (6.25,7.7) {};

\node[draw,circle,scale = 0.6, fill = white!50!black] (bx41) at (8.8,8.3) {};
\node[draw,circle,scale = 0.6, fill = white] (x42) at (9.0,7.9) {};

\node[draw,circle,scale = 0.6, fill = white!50!black] (bx51) at (10.7,8.3) {};
\node[draw,circle,scale = 0.6, fill = white] (x52) at (11.0,8.1) {};
\node[draw,circle,scale = 0.6, fill = white] (x53) at (10.5,7.7) {};


\node[draw,circle,scale = 0.7, fill = white] (C11) at (4.8,9.5) {};
\node[draw,circle,scale = 0.7, fill = white] (C12) at (6.5,9.5) {};
\node[draw,circle,scale = 0.7, fill = white] (C13) at (8.3,9.5) {};

\node[draw,circle,scale = 0.7, fill = white] (C21) at (2.7,10.5) {};
\node[draw,circle,scale = 0.7, fill = white] (C22) at (4.2,10.5) {};
\node[draw,circle,scale = 0.7, fill = white] (C23) at (10.3,10.5) {};

\node[draw,circle,scale = 0.7, fill = white] (C31) at (2.8,6.5) {};
\node[draw,circle,scale = 0.7, fill = white] (C32) at (4.5,6.5) {};
\node[draw,circle,scale = 0.7, fill = white] (C33) at (6.0,6.5) {};

\node[draw,circle,scale = 0.7, fill = white] (C41) at (7.0,6.5) {};
\node[draw,circle,scale = 0.7, fill = white] (C42) at (8.5,6.5) {};
\node[draw,circle,scale = 0.7, fill = white] (C43) at (10.0,6.5) {};

\node[draw,circle,scale = 0.7, fill = white] (C51) at (2.2,5.5) {};
\node[draw,circle,scale = 0.7, fill = white] (C52) at (6.5,5.5) {};
\node[draw,circle,scale = 0.7, fill = white] (C53) at (11.0,5.5) {};


\draw [color = black] (C11) -- (C12);
\draw [color = black] (C12) -- (C13);
\draw (C13) to [out=165,in=15] (C11);

\draw [color = black] (C21) -- (C22);
\draw [color = black] (C22) -- (C23);
\draw (C23) to [out=170,in=10] (C21);

\draw [color = black] (C31) -- (C32);
\draw [color = black] (C32) -- (C33);
\draw (C33) to [out=195,in=345] (C31);

\draw [color = black] (C41) -- (C42);
\draw [color = black] (C42) -- (C43);
\draw (C43) to [out=195,in=345] (C41);

\draw [color = black] (C51) -- (C52);
\draw [color = black] (C52) -- (C53);
\draw (C53) to [out=190,in=350] (C51);


\draw[color=black] (x11) -- (C21);
\draw[color=black] (x21) -- (C22);
\draw[color=black] (x22) -- (C11);
\draw[color=black] (x31) -- (C12);
\draw[color=black] (x41) -- (C13);
\draw[color=black] (x51) -- (C23);

\draw[color=black] (bx12) -- (C31);
\draw[color=black] (bx23) -- (C32);
\draw[color=black] (bx34) -- (C33);
\draw[color=black] (bx32) -- (C41);
\draw[color=black] (bx42) -- (C42);
\draw[color=black] (bx53) -- (C43);

\draw[color=black] (bx13) -- (C51);
\draw[color=black] (bx33) -- (C52);
\draw[color=black] (bx52) -- (C53);

\end{tikzpicture}}
\caption{Planar embedding in Mohar's style}
\label{subfig:first_step}
\end{subfigure}
\begin{subfigure}[b]{0.51\columnwidth}
\centering
\scalebox{0.6}{\begin{tikzpicture}



\draw [color = black] (2.0,7.7) -- (3.0,7.7) -- (3.0,8.3) -- (2.0,8.3) --  (2.0,7.7);
\draw [color = black] (4.0,7.7) -- (5.0,7.7) -- (5.0,8.3) -- (4.0,8.3) --  (4.0,7.7);
\draw [color = black] (6.0,7.7) -- (7.0,7.7) -- (7.0,8.3) -- (6.0,8.3) --  (6.0,7.7);
\draw [color = black] (8.0,7.7) -- (9.0,7.7) -- (9.0,8.3) -- (8.0,8.3) --  (8.0,7.7);
\draw [color = black] (10.0,7.7) -- (11.0,7.7) -- (11.0,8.3) -- (10.0,8.3) --  (10.0,7.7);


\node[draw,circle,scale = 0.6, fill = white] (x11) at (2.7,8.3) {};
\node[draw,circle,scale = 0.6, fill = white!50!black] (bx12) at (2.8,7.7) {};
\node[draw,circle,scale = 0.6, fill = white!50!black] (bx13) at (2.2,7.7) {};

\node[draw,circle,scale = 0.6, fill = white] (x21) at (4.2,8.3) {};
\node[draw,circle,scale = 0.6, fill = white] (x22) at (4.8,8.3) {};
\node[draw,circle,scale = 0.6, fill = white!50!black] (bx23) at (4.5,7.7) {};

\node[draw,circle,scale = 0.6, fill = white] (x31) at (6.5,8.3) {};
\node[draw,circle,scale = 0.6, fill = white!50!black] (bx32) at (7.0,7.7) {};
\node[draw,circle,scale = 0.6, fill = white!50!black] (bx33) at (6.5,7.7) {};
\node[draw,circle,scale = 0.6, fill = white!50!black] (bx34) at (6.0,7.7) {};

\node[draw,circle,scale = 0.6, fill = white] (x41) at (8.3,8.3) {};
\node[draw,circle,scale = 0.6, fill = white!50!black] (bx42) at (8.5,7.7) {};

\node[draw,circle,scale = 0.6, fill = white] (x51) at (10.3,8.3) {};
\node[draw,circle,scale = 0.6, fill = white!50!black] (bx52) at (11.0,7.7) {};
\node[draw,circle,scale = 0.6, fill = white!50!black] (bx53) at (10.0,7.7) {};


\node[draw,circle,scale = 0.6, fill = white!50!black] (bx11) at (3.0,8.3) {};
\node[draw,circle,scale = 0.6, fill = white] (x12) at (3.0,7.9) {};
\node[draw,circle,scale = 0.6, fill = white] (x13) at (2.5,7.7) {};

\node[draw,circle,scale = 0.6, fill = white!50!black] (bx21) at (4.5,8.3) {};
\node[draw,circle,scale = 0.6, fill = white!50!black] (bx22) at (5.0,8.0) {};
\node[draw,circle,scale = 0.6, fill = white] (x23) at (4.8,7.7) {};

\node[draw,circle,scale = 0.6, fill = white!50!black] (bx31) at (6.8,8.3) {};
\node[draw,circle,scale = 0.6, fill = white] (x32) at (7.0,8.0) {};
\node[draw,circle,scale = 0.6, fill = white] (x33) at (6.75,7.7) {};
\node[draw,circle,scale = 0.6, fill = white] (x34) at (6.25,7.7) {};

\node[draw,circle,scale = 0.6, fill = white!50!black] (bx41) at (8.8,8.3) {};
\node[draw,circle,scale = 0.6, fill = white] (x42) at (9.0,7.9) {};

\node[draw,circle,scale = 0.6, fill = white!50!black] (bx51) at (10.7,8.3) {};
\node[draw,circle,scale = 0.6, fill = white] (x52) at (11.0,8.1) {};
\node[draw,circle,scale = 0.6, fill = white] (x53) at (10.5,7.7) {};


\node[draw,circle,scale = 0.7, fill = white] (C11) at (4.8,9.5) {};
\node[draw,circle,scale = 0.7, fill = white] (C12) at (6.5,9.5) {};
\node[draw,circle,scale = 0.7, fill = white] (C13) at (8.3,9.5) {};

\node[draw,circle,scale = 0.7, fill = white] (C21) at (2.7,10.5) {};
\node[draw,circle,scale = 0.7, fill = white] (C22) at (4.2,10.5) {};
\node[draw,circle,scale = 0.7, fill = white] (C23) at (10.3,10.5) {};

\node[draw,circle,scale = 0.7, fill = white] (C31) at (2.8,6.5) {};
\node[draw,circle,scale = 0.7, fill = white] (C32) at (4.5,6.5) {};
\node[draw,circle,scale = 0.7, fill = white] (C33) at (6.0,6.5) {};

\node[draw,circle,scale = 0.7, fill = white] (C41) at (7.0,6.5) {};
\node[draw,circle,scale = 0.7, fill = white] (C42) at (8.5,6.5) {};
\node[draw,circle,scale = 0.7, fill = white] (C43) at (10.0,6.5) {};

\node[draw,circle,scale = 0.7, fill = white] (C51) at (2.2,5.5) {};
\node[draw,circle,scale = 0.7, fill = white] (C52) at (6.5,5.5) {};
\node[draw,circle,scale = 0.7, fill = white] (C53) at (11.0,5.5) {};


\draw [color = black] (C11) -- (C12);
\draw [color = black] (C12) -- (C13);
\draw (C13) to [out=165,in=15] (C11);

\draw [color = black] (C21) -- (C22);
\draw [color = black] (C22) -- (C23);
\draw (C23) to [out=170,in=10] (C21);

\draw [color = black] (C31) -- (C32);
\draw [color = black] (C32) -- (C33);
\draw (C33) to [out=195,in=345] (C31);

\draw [color = black] (C41) -- (C42);
\draw [color = black] (C42) -- (C43);
\draw (C43) to [out=195,in=345] (C41);

\draw [color = black] (C51) -- (C52);
\draw [color = black] (C52) -- (C53);
\draw (C53) to [out=190,in=350] (C51);


\draw[color=black] (x11) -- (C21);
\draw[color=black] (x21) -- (C22);
\draw[color=black] (x22) -- (C11);
\draw[color=black] (x31) -- (C12);
\draw[color=black] (x41) -- (C13);
\draw[color=black] (x51) -- (C23);

\draw[color=black] (bx12) -- (C31);
\draw[color=black] (bx23) -- (C32);
\draw[color=black] (bx34) -- (C33);
\draw[color=black] (bx32) -- (C41);
\draw[color=black] (bx42) -- (C42);
\draw[color=black] (bx53) -- (C43);

\draw[color=black] (bx13) -- (C51);
\draw[color=black] (bx33) -- (C52);
\draw[color=black] (bx52) -- (C53);


\draw[color = orange, line width = 1.4pt, dotted] (2.0,11.7) -- (11.0,11.7);
\node[scale = 1.2, color = orange] at (10.7,11.4) {$\mbox{x}^+$};
\draw[color = orange, line width = 1.4pt, dotted] (2.0,4.0) -- (11.0,4.0);
\node[scale = 1.2, color = orange] at (10.7,4.4) {$\mbox{x}^-$};

\node[draw,circle,scale = 0.7, fill = white] (P1) at (6.5,11.7) {};
\node[draw,circle,scale = 0.7, fill = white] (P2) at (4.2,11.7) {};
\node[draw,circle,scale = 0.7, fill = white] (P3) at (4.5,4.0) {};
\node[draw,circle,scale = 0.7, fill = white] (P4) at (8.5,4.0) {};
\node[draw,circle,scale = 0.7, fill = white] (P5) at (6.5,4.0) {};

\draw[color=red, line width = 1.4pt] (C11) -- (P1);
\draw[color=red, line width = 1.4pt] (C12) -- (P1);
\draw[color=red, line width = 1.4pt] (C13) -- (P1);

\draw[color=red, line width = 1.4pt] (C21) -- (P2);
\draw[color=red, line width = 1.4pt] (C22) -- (P2);
\draw[color=red, line width = 1.4pt] (C23) -- (P2);

\draw[color=red, line width = 1.4pt] (C31) -- (P3);
\draw[color=red, line width = 1.4pt] (C32) -- (P3);
\draw[color=red, line width = 1.4pt] (C33) -- (P3);

\draw[color=red, line width = 1.4pt] (C41) -- (P4);
\draw[color=red, line width = 1.4pt] (C42) -- (P4);
\draw[color=red, line width = 1.4pt] (C43) -- (P4);

\draw[color=red, line width = 1.4pt] (C51) -- (P5);
\draw[color=red, line width = 1.4pt] (C52) -- (P5);
\draw[color=red, line width = 1.4pt] (C53) -- (P5);

\end{tikzpicture}}
\caption{Embedding of $G_{\varphi}''$ (not planar)}
\label{subfig:first_step_pendant}
\end{subfigure}
\caption{Graph $G_{\varphi}''$ with and without pendant vertices.}
\label{fig:first_step}
\end{figure}

\begin{lemma}
Graph $G_{\varphi}''$ is 1-extendable.
\label{le:first_step_reduction}
\end{lemma}
\begin{proof}
Sets $\set{v_j^1,v_j^2,v_j^3,\pi_j}$, for any $1\le j\le m$, form an induced $K_4$. Hence there is a clique cover made up of all clauses $K_4$ together with half of the edges of the variable cycles. Hence an MIS of $G_{\varphi}''$ has size at most $m + c^*/2$. Here is an independent set of size $m + c^*/2$: pick up all pendant vertices $\pi_j$, $1\le j\le m$, and, for each variable, an MIS of the corresponding variable cycle (each cycle, being of even length, has two MISs, we may take any of them). This shows not only that $\alpha(G_{\varphi}'') = m + c^*/2$ but also that all pendant and cycle vertices belong to some MIS. Finally, each triangle vertex $v_j^q$ also belongs to an MIS: in the variable cycle adjacent to $v_j^q$, take the MIS not adjacent to $v_j^q$; in the other variable cycles, take any MIS in it; finally, take the pendant vertices of the other clauses ({\em i.e.} $\pi_h$ for $h \neq j$).
\end{proof}

\textbf{Construction of} $G_{\varphi}'$. The second step consists in transforming $G_{\varphi}''$ into some equivalent graph $G_{\varphi}'$ which is planar and has maximum degree 3. Two types of crossings appear in the embedding of $G_{\varphi}''$. Each of them necessarily involve pendant edges.

\begin{itemize}
\item \textbf{Type A}: a pendant edge $v_j^q\pi_j$ crosses a $T$-edge $v_{j'}^{p'}v_{j'}^{q'}$ (we may have $j = j'$),
\item \textbf{Type B}: a pendant edge $v_j^q\pi_j$ crosses another pendant edge $v_{j'}^{q'}\pi_{j'}$, $j\neq j'$.
\end{itemize}

We observe that for any of these types of crossings in the embedding of $G_{\varphi}''$, the assumptions of Proposition~\ref{prop:garey-johnson} are fulfilled.

\begin{lemma}
Let $\set{uu',vv'}$ be a crossing of the embedding of $G_{\varphi}''$. There exist two MISs $S_u,S_{u'}$ of $G_{\varphi}''$ which intersect $\set{u,u',v,v'}$ respectively in $\set{u}$ and $\set{u'}$.
\label{le:use_gjs}
\end{lemma}
\begin{proof}
We distinguish two cases, depending on the type of crossing.

\textit{Type A}. Let $uu' = v_j^q\pi_j$ and $vv' = v_{j'}^{p'}v_{j'}^{q'}$. Let $S_u$ be the MIS containing $v_j^q$ with all pendant vertices $\pi_h$, $h\neq j$ of other clauses (and $c^*/2$ cycle vertices selected properly). Let $S_{u'}$ be the MIS containing all pendant vertices (plus $c^*/2$ cycle vertices).

\textit{Type B}. Let $uu' = v_j^q\pi_j$ and $vv' = v_{j'}^{q'}\pi_{j'}$. We fix some $p' \in \set{1,2,3}$, $p' \neq q'$. Let $S_u$ be the MIS we finally obtain by picking up $v_j^q$, $v_{j'}^{p'}$, and all pendant vertices $\pi_h$, $h \neq j,j'$. There is no assignation conflict between $v_j^q$ and $v_{j'}^{p'}$ since only pendant edges involving literals of the same sign can cross each other with the monotone rectilinear representation. Similarly, let $S_{u'}$ be the MIS we obtain by taking $v_{j'}^{p'}$ and all pendant vertices $\pi_h$, $h \neq j'$, in particular $u' = \pi_j$.
\end{proof}

As a consequence of Lemma~\ref{le:use_gjs} together with Proposition~\ref{prop:garey-johnson}, one can replace each crossing of the embedding of $G_{\varphi}''$ by a gadget $\gjs$ without altering its $1$-extendability. The graph obtained is thus planar and has maximum degree 6 (which is the maximum degree of graph $\gjs$). Then, we apply transformation $(T_3)$ with $\Delta = 6$ to decrease its maximum degree. Finally, we obtain graph $G_{\varphi}'$, which is planar and has maximum degree 3.

\begin{lemma}
Graph $G_{\varphi}'$ is 1-extendable.
\label{le:second_step_reduction}
\end{lemma}
\begin{proof}
First, apply Lemma~\ref{le:use_gjs} together with Proposition~\ref{prop:garey-johnson}, second Lemma~\ref{le:t3}.
\end{proof}

According to Transformation $(T_3)$, each vertex $u$ of $G_{\varphi}''$ is transformed into an induced path $P_u$ of length 11 in $G_{\varphi}'$. There is a natural correspondence between the MISs of $G_{\varphi}''$ and those of $G_{\varphi}'$. Given an MIS $S''$ of $G_{\varphi}''$, one can produce an MIS $S'$ of $G_{\varphi}'$ such that the paths $P_u$ with $\card{P_u \cap S'} = 6$ represent the vertices $u \in S''$. Conversely, let $S'$ be some MIS of $G'$: picking up the vertices $u$ such that $\card{P_u \cap S'} = 6$ produces an MIS of $G_{\varphi}''$. In the remainder of the section, when we say that a vertex $u \in G_{\varphi}''$ belongs to an MIS $S'$ of $G_{\varphi}'$, we actually mean that $\card{P_u \cap S'} = 6$.

\textbf{Construction of} $G_{\varphi}$. We are now ready to describe the final graph $G_{\varphi}$ which consists in a small extension of $G_{\varphi}'$. We add a cycle $z_1,\bar{z}_1,\ldots,z_m,\bar{z}_m$ of size $2m$ to the graph $G_{\varphi}'$. Let $Z = \set{z_1,z_2,\ldots,z_m}$ and $\bar{Z} = \set{\bar{z}_1,\bar{z}_2,\ldots,\bar{z}_m}$. We connect $z_j$ to $\pi_j$ for every $1\le j\le m$ - concretely, as $\pi_j$ became an induced path via transformation $(T_3)$, we add an edge between $z_j$ and a vertex of $P_{\pi_j}$. The graph obtained is $G_{\varphi}$ and its size is polynomial in $|\varphi|$. The graph $G_{\varphi}$ is planar: consider the embedding of $G_{\varphi}'$, draw the cycle $Z \cup \bar{Z}$ as a rectangle surrounding it and such that all edges $\pi_jz_j$ are vertical. Its maximum degree is 3. We are now ready to prove our result.

\begin{theorem}
\textsc{$1$-Extendability}  is NP-hard, even on planar graphs of maximum degree $3$.
\label{th:hard_planar}
\end{theorem}
\begin{proof}
We begin with the proof that $\alpha(G_{\varphi})=\alpha(G_{\varphi}') + m$. The value $\alpha(G_{\varphi}') + m$ is clearly an upper bound of the size of independent sets in $G_{\varphi}$ as they cannot contain $m+1$ elements of $Z \cup \bar{Z}$. Moreover, the union of any MIS of $G_{\varphi}'$ with $\bar{Z}$ is an independent set of this size. 
As a consequence, for all $S \subseteq V(G_{\varphi})$, $S$ is an MIS of $G_{\varphi}$ iff $S \cap V(G'_{\varphi})$ is an MIS of $G'_{\varphi}$ and $S \cap (Z \cup \bar{Z})$ is an MIS of $G[Z \cup \bar{Z}]$.
As $G_{\varphi}'$ is 1-extendable (Lemma~\ref{le:second_step_reduction}), we know that all vertices of $V(G_{\varphi}')$ belong to some MIS of $G_{\varphi}$.
Moreover $G[Z \cup \bar{Z}]$ contains exactly two MIS: $Z$ and $\bar{Z}$.  As said previously, every vertex of $\bar{Z}$ is contained in an MIS of $G_{\varphi}$. Hence, $G_{\varphi}$ is $1$-extendable iff it admits an MIS  containing $Z$.

Assume $\varphi$ is satisfiable: there is an assignment $A$ of variables which satisfies $\varphi$. We describe an MIS of $G''_{\varphi}$, which can easily be transformed into an MIS of $G'_{\varphi}$. If $x_i$'s assignment via $A$ is \textsf{True}, we pick vertices $\bar{x}_i^1,\bar{x}_i^2,\ldots$ of its variable cycle, otherwise we pick the other parity $x_i^1,x_i^2,\ldots$. Therefore, if $x_i$'s assignment is \textsf{True}, we cannot pick any triangle vertex representing $\neg x_i$ as it is in conflict with one $\bar{x}_i^s$. This assignment $A$ is such that at least one literal of each clause is assigned to \textsf{True}. So, for each clause $C_j$, we select arbitrary one of its literals $\ell_j^q$ which is positively assigned with $A$ and pick the vertex $v_j^q$. The chosen vertices form an independent set $S_A''$ of size $\alpha(G_{\varphi}'')$ in $G_{\varphi}''$. Observe that $S_A''$ does not contain any pendant vertex. 
This set $S_A''$ can be transformed into a corresponding MIS $S_A'$ in $G_{\varphi}'$ which contains 6 vertices per path $P_u$ if $u \in S_A''$. Thus, no pendant $\pi_j$ belongs to $S_A'$. As the neighborhood of $Z$ in $G_{\varphi}$ is made up only of pendants $\pi_j$, $S_A' \cup Z$ is an MIS of $G_{\varphi}$ and the graph is 1-extendable.

Suppose now that $G_{\varphi}$ is 1-extendable. So there is an MIS $Z \cup S'$ of $G_{\varphi}$, with $S' \subseteq V(G'_{\varphi})$. We know that, from $S'$, we can retrieve an MIS $S''$ of $G_{\varphi}''$ which contains vertices $u$ such that $\card{P_u \cap S'} = 6$. The set $S''$ cannot contain pendant vertices as they are all ``adjacent'' to $Z$ in $G_{\varphi}$. We propose the following variable assignment $A$. Half of the vertices of the variable cycles must be in $S''$, otherwise it would not be an MIS. Consequently, if $x_i^1,x_i^2,\ldots$ belong to $S''$, we assign $x_i$ to \textsf{False}, otherwise to \textsf{True}. Let us check that all clauses are satisfied. The set $S''$ contains exactly one vertex per triangle representing $C_j$. This vertex must be in accordance with the parity of the variable cycle which is in $S''$: for example, if $\ell_j^q = \neg x_i$ and $v_j^q \in S''$, then we have $x_i^1,x_i^2,\ldots \in S''$, otherwise $S''$ would not be independent. Hence, $S''$ cannot contain two vertices $v_j^q$ and $v_{j'}^{q'}$ such that $\ell_j^q = \neg \ell_{j'}^{q'}$. In summary, assignment $A$ satisfies $\varphi$.
\end{proof}

As for \textsc{Maximum Independent Set}, the problem \textsc{$1$-Extendability} is NP-hard on subcubic planar graphs. If we put aside the degree criterion, one can see that this proof also works if we do not use transformation $(T_3)$. However, it stays relatively tricky, while the NP-hardness of \textsc{Maximum Independent Set} for planar graphs consists only in replacing each crossing of an arbitrary embedding of $G$ by the GJS-gadget. Unfortunately, as mentioned in Section~\ref{subsec:gadget}, such reduction does not work for \textsc{$1$-Extendability}. We wonder whether a new gadget, certainly not so much different from $\gjs$, could be designed to make this reduction simpler.

\subsection{Unit disk graphs} \label{subsec:unit_disk}

Unit disk graphs~\cite{ClCoJo90} stand as a natural model for wireless networks. Indeed, they are defined as the intersection graph of $n$ equal-sized disks in the plane, which can represent Wi-Fi access points with the same radio range. There exists a way to represent subdivided planar graphs with degree at most 4 as unit disk graphs, based on a result from Valiant~\cite{Va81}.

\begin{theorem}[\cite{Va81}]
A planar graph $G$ with maximum degree 4 can be embedded in the plane inside a $O(|V(G)|)$-sized area in such a way that any vertex is at integer coordinates and each edge is made up of vertical and horizontal line segments.
\end{theorem}

Consider a subcubic planar graph $G$ with such an embedding. We subdivide the edges of $G$ such that (i) a vertex is placed at each turn of every edge (ii) each segment (between two turns) is subdivided at least once and (iii) the distance between two adjacent vertices is at most half of the length of the shortest segment. The graph obtained - say $G^{\mbox{\scriptsize{UD}}}$ - is unit disk. Indeed, each vertex - admitting at most four neighbors, each of them placed either in the two horizontal or vertical directions - can be represented as a disk $\mathcal{D}$ with center $d$ such that:

\begin{itemize}
\item it intersects a disk $\mathcal{D}_{\mbox{\scriptsize{N}}}$ (resp. $\mathcal{D}_{\mbox{\scriptsize{S}}}$, $\mathcal{D}_{\mbox{\scriptsize{W}}}$, $\mathcal{D}_{\mbox{\scriptsize{E}}}$) with center $d_{\mbox{\scriptsize{N}}}$ (resp. $d_{\mbox{\scriptsize{S}}}$, $d_{\mbox{\scriptsize{W}}}$, $d_{\mbox{\scriptsize{E}}}$) where arc $\overrightarrow{dd_{\mbox{\scriptsize{N}}}}$ (resp. $\overrightarrow{dd_{\mbox{\scriptsize{S}}}}$, $\overrightarrow{dd_{\mbox{\scriptsize{W}}}}$, $\overrightarrow{dd_{\mbox{\scriptsize{E}}}}$) indicates the North direction (resp. South, West, East)
\item disks $\mathcal{D}_{\mbox{\scriptsize{N}}}$, $\mathcal{D}_{\mbox{\scriptsize{S}}}$, $\mathcal{D}_{\mbox{\scriptsize{W}}}$, $\mathcal{D}_{\mbox{\scriptsize{E}}}$ do not intersect each other
\end{itemize}

In summary, given a subcubic planar graph $G$, one can produce in polynomial-time a unit disk graph $G^{\mbox{\scriptsize{UD}}}$ thanks to transformation $(T_2)$. According to Lemma~\ref{le:t2}, $G$ is 1-extendable iff $G^{\mbox{\scriptsize{UD}}}$ is 1-extendable.

\begin{theorem}
\textsc{$1$-Extendability}  is NP-hard, even on unit disk graphs.
\label{th:unit_disk}
\end{theorem}

\section{Parameterized algorithms}\label{sec:param}

In this section we study a parameterized version of the $1$-extendability problem:

\defparproblem{param-$1$-Extendability}{A graph $G$, an integer $k$}{$k$}{Does every vertex of $G$ belong to an independent set of size $k$?}

We first show that the problem remains $W[1]$-hard by a reduction from \textsc{Multicolored Independent Set}. We then investigate the existence of polynomial kernels in restricted graph classes.

\begin{theorem}\label{thm:w1}
\textsc{param-$1$-Extendability} is $W[1]$-hard.
\end{theorem}
\begin{proof}
We reduce from \textsc{Multicolored Independent Set}, where the input is a graph $G$ whose vertex set is partitioned into $k$ cliques $C_1$, $\dots$, $C_k$, and the goal is to find an independent set of size $k$. We add, for every $i \in [k]$, a pendant vertex $\pi_i$ adjacent to all vertices of $C_i$. We also add a vertex $\omega$ adjacent to $\pi_i$, $i \in [k]$, and a vertex $\pi_{\omega}$ adjacent to $\omega$ only. Let $G'$ be the obtained graph.
We claim that every vertex of $G'$ is contained in an independent set of size $k+1$ iff $G$ contains an independent set of size $k$.
If $G$ contains an independent set $S$ of size $k$, then:
\begin{itemize}
	\item for every $i \in [k]$, every $x \in C_i$, then $x$ together with $\{\pi_j: j \neq i\} \cup \{\pi_{\omega}\}$ is an independent set of size $k+1$
	\item $S$ together with $\omega$ is an independent set of size $k+1$.
	\item the set $\{\pi_i : i \in [k]\} \cup \{\pi_{\omega}\}$ is an independent set of size $k+1$.
\end{itemize}
Conversely, assume $\omega$ is in an independent set of size $k+1$. Then, since the set of non-neighbors of $\omega$ is $G$ itself, it implies that $G$ must contain an independent set of size $k$, which concludes the proof.
\end{proof}

By using Lemma~\ref{lem:1ext-to-alpha}, the problem is FPT in every hereditary graph class where \textsc{Maximum Independent Set} is FPT. Examples of such classes are planar graphs and triangle-free graphs (and more generally graphs excluding a clique of size $r$ as an induced subgraph, for every fixed $r \geqslant 3$) where, in addition, \textsc{Maximum Independent Set} admits a kernel of polynomial size. The reduction of Lemma~\ref{lem:1ext-to-alpha}, however, does not preserve polynomial kernels (although it naturally gives polynomial \textit{Turing} kernels).
 Hence, it is natural to ask whether \textsc{param-$1$-Extendability} admits a polynomial kernel in these classes. We answer positively to this question.

We say that a hereditary graph class $\C$ is \textit{\textsc{MIS}-$(c, t)$-friendly}, for two non-zero constants $c$ and $t$, if every graph of the class on $n$ vertices contains an independent set of size at least $t \cdot n^c$, and such an independent set can be found in polynomial-time.

\begin{theorem}\label{thm:polyker}
Let $\C$ be an \textit{\textsc{MIS}-$(c, t)$-friendly} class. \textsc{param-$1$-Extendability} on $\C$ admits a kernel with $O(k^{\frac{1}{c}+\frac{1}{c^2}})$ vertices. 
\end{theorem}
\begin{proof}
Let $G \in \C$, and $k \in \mathbb{N}$. Let $\A$ be the algorithm which, for every graph of $\C$ on $n$ vertices, returns an independent set of size $t \cdot n^c$.
We assume $|V(G)| \geqslant \left(\frac{k}{t}\right)^{\frac{1}{c}}$, since otherwise we are done. We invoke $\A$ on $G$ in order to get an independent set $S_0$, which is thus of size at least $k$. We remove $S_0$ and repeat the process on the remaining vertices until $\A$ outputs a small independent set. More precisely, let $R_1 = V(G) \setminus S_0$, and start with $i=1$. We run $\A$ on $G[R_i]$ which outputs an independent set $S_i$. If $|S_i| < k$, we stop the process, and otherwise we continue with $R_{i+1} = R_i \setminus S_i$, and increment $i$.
Eventually, we end up with a partition of $V(G)$ into $S_0$, $S_1$, $\cdots$, $S_q$ and $R_{q+1}$ (we might have $q=0$). Every $S_i$ is an independent set of size at least $k$ in $G$. As $G[R_{q+1}] \in \C$, algorithm $\A$ produces an independent set of size $t\card{R_{q+1}}^c < k$ on it, so $|R_{q+1}| < \left(\frac{k}{t}\right)^{\frac{1}{c}}$.

We now describe a reduction rule which consists of a marking procedure of some vertices of $S_0$, and removing those which were not marked. We then show that if the reduction rule does not remove any vertex, then it means that the graph has the desired number of vertices.

\textit{Marking procedure.} 
First, mark $k$ vertices of $S_0$ chosen arbitrarily. For every $x \in R_{q+1}$, let $s_x$ be the number of non-neighbors of $x$ in $S_0$. Second, for each $x \in R_{q+1}$, mark $\min \{k-1, s_x\}$ vertices of $S_0$ chosen arbitrarily.
As announced, we remove all vertices of $S_0$ which were not marked by the previous procedure.

\textit{Safeness.} Let $G'$ be the graph obtained after the reduction rule, and $S_0'$ the vertices of $V(G') \cap S_0$. 
Suppose every vertex of $G'$ belongs to an independent set of size $k$. We only need to show that every removed vertex is in an independent set of size $k$ in $G$. This is indeed the case, as we only removed vertices from $S_0$ and we kept at least $k$ vertices from $S_0$.
Conversely, suppose that every vertex of $G$ is in an independent set of size $k$. Then:
\begin{itemize}
	\item since $|S_0'| \geqslant k$ (as we marked $k$ vertices from $S_0$), every vertex of $S_0'$ belongs to an independent set of size $k$.
	\item For every $i \in [q]$, since each $S_i$ is of size at least $k$, every vertex of each $S_i$ is in an independent set of size at least $k$
	\item for every $x \in R_{q+1}$, let $S$ be an independent set of size $k$ in $G$ containing $x$. Since $|S \cap S_0| \leqslant k-1$, we necessarily marked at least $|S \cap S_0|$ non-neighbors of $x$ in $S_0$, hence we can always replace the removed vertices of $S$ by other vertices of $S_0'$ so that $x$ belongs to an independent set of size $k$ in $G'$.
\end{itemize}

\paragraph{Size of the reduced instance.} We apply the reduction rule as long as we can. Since we remove at least one vertex if the reduction rule applies, the algorithm must end after $O(|V(G)|)$ applications of the rule. 
Then, if the reduction rule cannot apply, it means we mark all vertices of $S_0$. But since we mark at most $k + (k-1) \cdot \left(\frac{k}{t}\right)^{\frac{1}{c}}$ vertices of $S_0$, and $S_0$ is of size at least $t \cdot |V(G)|^c$, it means that $G$ has $O(k^{\frac{1}{c^2}+\frac{1}{c}})$ vertices.
This concludes the proof of the theorem. 
\end{proof}

We now apply the previous theorem to planar graphs and $K_r$-free graphs. By the Four Color Theorem, every planar graph on $n$ vertices contains an independent set of size $n/4$ which can be found in polynomial-time. Hence, planar graphs is an \textsc{MIS}-$(1, 1/4)$-friendly class. More generally, $d$-degenerate graphs are MIS-$(1,\frac{1}{d+1})$-friendly. By Ramsey's theorem, for every $r \geqslant 3$, every $K_r$-free graphs on $n$ vertices contains an independent set of size $n^{\frac{1}{r-1}}$ which can be found in polynomial-time. Hence, $K_r$-free graphs is an \textsc{MIS}-$\left(\frac{1}{r-1}, 1\right)$-friendly class.
\begin{corollary}\label{cor:planar}
\textsc{param-$1$-Extendability} admits a kernel with $O(k^2)$ vertices on planar graphs and $d$-degenerate graphs for bounded $d$, and a kernel with $O(k^{r^2})$ vertices on $K_r$-free graphs for every fixed $r \geqslant 3$.
\label{co:kernel}
\end{corollary}

\section{Conclusion and further research} \label{sec:conclusion}

We investigated the computational complexity of \textsc{$1$-Extendability}. We showed that in general graphs it cannot be solved in subexponential-time unless the ETH fails, and that it remains NP-hard in subcubic planar graphs and in unit disk graphs. Although this behavior seems to be the same as \textsc{Maximum Independent Set}, we proved that \textsc{Maximum Independent Set} remains NP-hard (and even W[1]-hard) in $1$-extendable graphs.
It seems challenging to find a larger class of graphs where \textsc{$1$-Extendability} is polynomial-time solvable (but not trivial) while \textsc{Maximum Independent Set} remains NP-hard.

Another interesting subject would be to characterize $1$-extendable graphs of graph classes where \textsc{Maximum Independent Set} is polynomial-time solvable: \textit{e.g.} chordal graphs, cographs, claw-free graphs. Such outcomes would extend the result of Dean and Zito~\cite{DeZi94} which state that bipartite graphs are 1-extendable iff they admit a perfect matching. 

We also studied \textsc{param-$1$-Extendability}, a parameterized version of \textsc{$1$-Extendability} and showed that some results for \textsc{Maximum Independent Set} could also be obtained for \textsc{param-$1$-Extendability} (although not being as direct). 
It would be interesting to determine whether this is also the case for other results about \textsc{Maximum Independent Set}~\cite{BoBoChThWa20,BoBoThWa19,DaLoMuRa12}, for instance: is \textsc{param-$1$-Extendability} W[1]-hard in $C_4$-free graphs and in $K_{1,4}$-free graphs? Does it admit a polynomial kernel in diamond-free graphs?

Finally, because of its applications in network design, finding an efficient algorithm which works well in practice is of high importance. Toward this, a first step would be to determine in which cases a vertex addition (or deletion) preserves the property of being $1$-extendable. We note that such results have already been obtained for the related property of being well-covered~\cite{FiWh18}.

 \bibliographystyle{splncs04}
%
\bibliography{bib_extendability}

\end{document}